\documentclass[11pt,twoside,a4paper,notitlepage]{article}

\usepackage{ifthen}

\newboolean{publisherversion}
\setboolean{publisherversion}{false}

\ifthenelse{\boolean{publisherversion}}
{}
{ 
  
  \usepackage[letterpaper,margin=1in]{geometry}
  \usepackage{tgtermes}
  \usepackage[scale=.92]{tgheros}
  \usepackage{tgcursor}
}

\ifthenelse{\boolean{publisherversion}}
{}
{
}

\ifthenelse{\boolean{publisherversion}}
{}
{
  \usepackage[english]{babel}
}

\ifthenelse{\boolean{publisherversion}}
{}
{
  \usepackage{fancyhdr}
}

\ifthenelse{\boolean{publisherversion}}
{}
{
  
}

\usepackage{ifthen}

\usepackage{ifthen}

\provideboolean{detectedSTOC}
\provideboolean{detectedFOCS}
\provideboolean{detectedElsevier}
\provideboolean{detectedNOW}
\provideboolean{detectedLMCS}
\provideboolean{detectedIEEE}
\provideboolean{detectedPoster}
\provideboolean{detectedSIAM}
\provideboolean{detectedLNCS}
\provideboolean{detectedACM}
\provideboolean{detectedACMconf}
\provideboolean{detectedSigplanconf}
\provideboolean{detectedToC}
\provideboolean{detectedLIPIcs}
\provideboolean{detectedAAAI}
\provideboolean{detectedIJCAI}
\provideboolean{detectedCompCplx}
\provideboolean{detectedEasyChair}
\provideboolean{detectedJAIR}
\provideboolean{detectedArticle}
\provideboolean{detectedReport}
\provideboolean{detectedThesis}

\makeatletter

\@ifclassloaded{sig-alternate}
{\setboolean{detectedSTOC}{true}}
{\setboolean{detectedSTOC}{false}}

\@ifclassloaded{elsarticle}
{\setboolean{detectedElsevier}{true}}
{\setboolean{detectedElsevier}{false}}

\@ifclassloaded{now}
{\setboolean{detectedNOW}{true}}
{\setboolean{detectedNOW}{false}}

\@ifclassloaded{lmcs}
{\setboolean{detectedLMCS}{true}}
{\setboolean{detectedLMCS}{false}}

\@ifclassloaded{IEEEtran} {
  \setboolean{detectedIEEE}{true}
  \ifCLASSOPTIONconference {
    \setboolean{detectedFOCS}{true}
  }
  \else {
    \setboolean{detectedFOCS}{false}
  }
  \fi
}
{
  \setboolean{detectedFOCS}{false}
  \setboolean{detectedIEEE}{false}
}

\@ifclassloaded{siamart171218}
{\setboolean{detectedSIAM}{true}}
{\setboolean{detectedSIAM}{false}}

\@ifclassloaded{llncs}
{\setboolean{detectedLNCS}{true}}
{\setboolean{detectedLNCS}{false}}

\@ifclassloaded{acmsmall}
{\setboolean{detectedACM}{true}}
{\setboolean{detectedACM}{false}}

\@ifclassloaded{acmart}
{\setboolean{detectedACMconf}{true}}
{\setboolean{detectedACMconf}{false}}

\@ifclassloaded{sigplanconf}
{\setboolean{detectedSigplanconf}{true}}
{\setboolean{detectedSigplanconf}{false}}

\@ifclassloaded{toc}
{\setboolean{detectedToC}{true}}
{\setboolean{detectedToC}{false}}

\@ifclassloaded{lipics}
{\setboolean{detectedLIPIcs}{true}}
{\@ifclassloaded{lipics-v2019}
  {\setboolean{detectedLIPIcs}{true}}
  {\@ifclassloaded{oasics-v2019}
    {\setboolean{detectedLIPIcs}{true}}
    {\setboolean{detectedLIPIcs}{false}}
  }
}

\@ifclassloaded{cc}
{\setboolean{detectedCompCplx}{true}}
{\setboolean{detectedCompCplx}{false}}

\@ifclassloaded{easychair}
{\setboolean{detectedEasyChair}{true}}
{\setboolean{detectedEasyChair}{false}}

\@ifpackageloaded{jair}
{\setboolean{detectedJAIR}{true}}        
{\setboolean{detectedJAIR}{false}}        

\@ifpackageloaded{aaai}
{\setboolean{detectedAAAI}{true}}        
{\@ifpackageloaded{aaai18}
  {\setboolean{detectedAAAI}{true}}
  {\@ifpackageloaded{aaai20}       
    {\setboolean{detectedAAAI}{true}}
    {\setboolean{detectedAAAI}{false}}}}

\@ifpackageloaded{ijcai18}
{\setboolean{detectedIJCAI}{true}}        
{\@ifpackageloaded{ijcai19}
  {\setboolean{detectedIJCAI}{true}}        
  {\setboolean{detectedIJCAI}{false}}}

\@ifclassloaded{sciposter}
{\setboolean{detectedPoster}{true}}
{\setboolean{detectedPoster}{false}}

\@ifclassloaded{article}
{\setboolean{detectedArticle}{true}}
{\setboolean{detectedArticle}{false}}

\makeatother

\ifthenelse{\not \isundefined{\examen} 
  \and \not \isundefined{\disputationsdatum} 
  \and \not \isundefined{\disputationslokal}}   
  {\setboolean{detectedThesis}{true}}
  {\setboolean{detectedThesis}{false}}

\ifthenelse{\boolean{detectedArticle} \or \boolean{detectedThesis}
  \or \boolean{detectedSTOC} \or \boolean{detectedFOCS}
  \or \boolean{detectedSIAM} \or \boolean{detectedIEEE}
  \or \boolean{detectedPoster}}
{\setboolean{detectedReport}{false}}
{\setboolean{detectedReport}{true}}

\ifthenelse{\boolean{detectedLIPIcs}}
{}
{\usepackage[utf8]{inputenc}}

\ifthenelse{\boolean{detectedIEEE}}
{\usepackage[cmex10]{amsmath}
  \interdisplaylinepenalty=2500}
{\usepackage{amsmath}}
\usepackage{amssymb}
\usepackage{amsfonts}

\usepackage{mathtools}

\usepackage{bussproofs}

\usepackage{algorithm}
\usepackage[noend]{algpseudocode}

\usepackage{microtype}

\usepackage{bm}  

\usepackage{subcaption}

\ifthenelse
{     \boolean{detectedElsevier} \or \boolean{detectedSIAM}
  \or \boolean{detectedSIAM}     \or \boolean{detectedLIPIcs}}
{}
{\usepackage{varioref}}

\usepackage{xspace}

\ifthenelse    
{\boolean{detectedSTOC}      \or \boolean{detectedSIAM}         \or 
  \boolean{detectedLMCS}     \or \boolean{detectedNOW}          \or 
  \boolean{detectedLNCS}     \or \boolean{detectedACM}          \or
  \boolean{detectedToC}      \or \boolean{detectedLIPIcs}       \or
  \boolean{detectedACMconf}  \or \boolean{detectedAAAI}         \or
  \boolean{detectedCompCplx} \or \boolean{detectedSigplanconf}}
{}
{
  \usepackage{amsthm}
  \usepackage{amsthm-boldfont}
}

\ifthenelse        
{\boolean{detectedElsevier}  \or \boolean{detectedFOCS}         \or 
  \boolean{detectedSTOC}     \or \boolean{detectedPoster}       \or
  \boolean{detectedSIAM}     \or \boolean{detectedLMCS}         \or
  \boolean{detectedIEEE}     \or \boolean{detectedNOW}          \or
  \boolean{detectedToC}      \or \boolean{detectedThesis}       \or
  \boolean{detectedLNCS}     \or \boolean{detectedACM}          \or 
  \boolean{detectedLIPIcs}   \or \boolean{detectedAAAI}         \or
  \boolean{detectedACMconf}  \or \boolean{detectedIJCAI}        \or 
  \boolean{detectedCompCplx} \or \boolean{detectedSigplanconf}}
{}
{\usepackage{jncaptionsheadings}}

\usepackage{ifthen}
\usepackage{xspace}
\ifthenelse
{\boolean{detectedElsevier} \or \boolean{detectedSIAM} 
  \or \boolean{detectedLIPIcs}}
{}
{\usepackage{varioref}}

\DeclareMathAlphabet{\mathsfsl}{OT1}{cmss}{m}{sl}

\DeclareRobustCommand{\BibTeX}{{\normalfont B\kern-.05em{\scshape i\kern-.025em b}\kern-.08em \TeX}}

\newcommand{\formuladots}{\cdots}

\newcommand{\Bigoh}[1]{\mathrm{O} \bigl( #1 \bigr)}
\newcommand{\bigoh}[1]{\mathrm{O} ( #1 )}

\newcommand{\littleoh}[1]{\mathrm{o} ( #1 )}

\newcommand{\Bigomega}[1]{\Omega \bigl( #1 \bigr)}

\newcommand{\littleomega}[1]{\omega ( #1 )}

\ifthenelse{\boolean{detectedToC}}{}
{
  
  \newcommand{\R}         {\mathbb{R}}
  
  \newcommand{\N}         {\mathbb{N}}

}

\newcommand{\MAXOFEXPR}[2][]{\max_{#1} \left\{ #2 \right\}}
\newcommand{\MINOFEXPR}[2][]{\min_{#1} \left\{ #2 \right\}}
\newcommand{\Maxofexpr}[2][]{\max_{#1} \bigl\{ #2 \bigr\}}
\newcommand{\Minofexpr}[2][]{\min_{#1} \bigl\{ #2 \bigr\}}

\newcommand{\MAXOFSET}[3][:]{\ifthenelse{\equal{#1}{;}}{\MAXOFEXPR{ #2 \,;\, #3 }}
     {\ifthenelse{\equal{#1}{:}}{\MAXOFEXPR{ #2 \,:\, #3 }}
     {\max \twincommandJN{\left\{}{#2}{\left#1}{\right}{\,#3}{\right\}}}}}
\newcommand{\MINOFSET}[3][:]{\ifthenelse{\equal{#1}{;}}{\MINOFEXPR{ #2 \,;\, #3 }}
     {\ifthenelse{\equal{#1}{:}}{\MINOFEXPR{ #2 \,:\, #3 }}
     {\min \twincommandJN{\left\{}{#2}{\left#1}{\right}{\,#3}{\right\}}}}}

\newcommand{\Maxofset}[3][:]{\ifthenelse{\equal{#1}{;}}{\Maxofexpr{ #2 \,;\, #3 }}
     {\ifthenelse{\equal{#1}{:}}{\Maxofexpr{ #2 \,:\, #3 }}
     {\max \twincommandJN{\bigl\{}{#2}{\bigl#1}{\bigr}{\,#3}{\bigr\}}}}}
\newcommand{\Minofset}[3][:]{\ifthenelse{\equal{#1}{;}}{\Minofexpr{ #2 \,;\, #3 }}
     {\ifthenelse{\equal{#1}{:}}{\Minofexpr{ #2 \,:\, #3 }}
     {\min \twincommandJN{\bigl\{}{#2}{\bigl#1}{\bigr}{\,#3}{\bigr\}}}}}

\newcommand{\F}{\mathbb{F}}

\ifthenelse{\boolean{detectedLMCS}}
{}
{}

\newcommand{\twincommandJN}[6]{#1#2#3\vphantom{#2#5}\mspace{-2.05mu}#4.#5#6}

\newcommand{\set}[1]{\{ #1 \}}
\newcommand{\Set}[1]{\bigl\{ #1 \bigr\}}

\newcommand{\setdescr}[3][\mid]{\set{ #2 #1 #3 }}
\newcommand{\Setdescr}[3][|]{\ifthenelse{\equal{#1}{;}}{\Set{ #2 \,;\, #3 }}
     {\ifthenelse{\equal{#1}{:}}{\Set{ #2 \,:\, #3 }}
     {\twincommandJN{\bigl\{}{#2\,}{\bigl#1}{\bigr}{\,#3}{\bigr\}}}}}

\newcommand{\setsize}[1]{\lvert#1\rvert}

\newcommand{\intersection}{\cap}

\newcommand{\union}{\cup}
\newcommand{\Union}{\bigcup}

\newcommand{\olnot}[1]{\overline{#1}}

\newcommand{\nvar}{n}
\newcommand{\nvars}{\nvar}

\newcommand{\complclassformat}[1]{\textrm{\upshape{\textsf{#1}}}\xspace}

\newcommand{\NP}{\complclassformat{NP}}

\newcommand{\introduceterm}[1]{{\emph{#1}}}

\newcommand{\eqperiod}{\enspace .}
\newcommand{\eqcomma}{\enspace ,}
\renewcommand{\eqperiod}{\, .}
\renewcommand{\eqcomma}{\, ,}

\newcommand{\wrt}{with respect to\xspace}

\ifthenelse{\boolean{detectedToC}}{}
  {\newcommand{\eg}{for instance\xspace} }
\ifthenelse{\boolean{detectedToC}}{}{
\newcommand{\ie}{i.e.,\ }

}
  
\ifthenelse{\boolean{detectedLIPIcs} \or \boolean{detectedIJCAI}}
{\renewcommand{\st}{\errmessage{Please do not use st}}}
{\ifthenelse{\isundefined{\st}}
  {\newcommand{\st}{such that\xspace}}
  {}}      

\newcommand{\etal}{et al.\@\xspace}

\ifthenelse{\boolean{detectedIEEE}}{}{}

\newcommand{\aas}{asymptotically almost surely\xspace}

\newcommand{\refsec}[1]{Section~\ref{#1}}

\newcommand{\refapp}[1]{Appendix~\ref{#1}}

\newcommand{\reffig}[1]{Figure~\ref{#1}}

\newcommand{\refth}[1]{Theorem~\ref{#1}}

\newcommand{\refthm}[1]{Theorem~\ref{#1}}

\newcommand{\reflem}[1]{Lemma~\ref{#1}}

\newcommand{\refdef}[1]{Definition~\ref{#1}}

\newcommand{\refobs}[1]{Observation~\ref{#1}}

\newcommand{\refclaim}[1]{Claim~\ref{#1}}

\newcommand{\refitem}[1]{item~\ref{#1}}

\ifthenelse
{\isundefined{\refeq}}
{\newcommand{\refeq}[1]{\eqref{#1}}}
{\renewcommand{\refeq}[1]{\eqref{#1}}}

\ifthenelse{\isundefined{\DONOTLOADHYPERREFPACKAGE}
  \and \not \boolean{detectedSTOC}        \and \not \boolean{detectedFOCS}
  \and \not \boolean{detectedPoster}      \and \not \boolean{detectedElsevier} 
  \and \not \boolean{detectedSIAM}        \and \not \boolean{detectedACM}
  \and \not \boolean{detectedIEEE}        \and \not \boolean{detectedNOW}
  \and \not \boolean{detectedToC}         \and \not \boolean{detectedThesis}
  \and \not \boolean{detectedLIPIcs}      \and \not \boolean{detectedSIAM}
  \and \not \boolean{detectedAAAI}        \and \not \boolean{detectedIJCAI}
  \and \not \boolean{detectedSigplanconf} \and \not \boolean{detectedACMconf}   
  \and \not \boolean{detectedCompCplx} \and \not \boolean{detectedEasyChair}}
{\usepackage[colorlinks=true,citecolor=blue]{hyperref}}
{}

\ifthenelse{\boolean{detectedSTOC}}                  
{

\newtheorem{theorem}{Theorem}
\newtheorem{lemma}[theorem]{Lemma}
\newtheorem{proposition}[theorem]{Proposition}
\newtheorem{corollary}[theorem]{Corollary}
\newtheorem{observation}[theorem]{Observation}

\newtheorem{definition}[theorem]{Definition}
\newtheorem{claim}[theorem]{Claim}

\newtheorem{conjecture}{Conjecture}
\newtheorem{openproblem}[conjecture]{Open Problem}

\newdef{remark}{Remark}
\newdef{example}{Example}

\newcounter{unnumber}

 }
{\ifthenelse{\boolean{detectedSIAM}}
  {

\newsiamremark{remark}{Remark}
\newsiamremark{hypothesis}{Hypothesis}

\newsiamthm{claim}{Claim}
\newsiamthm{observation}{Observation}
\newsiamthm{conjecture}{Conjecture}

\crefname{hypothesis}{Hypothesis}{Hypotheses}

\crefname{section}{Section}{Sections}
\crefname{claim}{Claim}{Claims}

\renewcommand{\refsec}[1]{\cref{#1}}

\renewcommand{\refapp}[1]{\cref{#1}}

\renewcommand{\reffig}[1]{\cref{#1}}

\renewcommand{\refth}[1]{\cref{#1}}

\renewcommand{\refthm}[1]{\cref{#1}}

\renewcommand{\reflem}[1]{\cref{#1}}

\renewcommand{\refdef}[1]{\cref{#1}}

\renewcommand{\refobs}[1]{\cref{#1}}

\newcommand{\refclaim}[1]{\cref{#1}}

 }
  {\ifthenelse{\boolean{detectedNOW}}
    {

\newtheorem{standardlocalcounter}{Dummy}[chapter]
\newtheorem{standardglobalcounter}{Dummy}

\newtheorem{theorem}[standardlocalcounter]{Theorem}
\newtheorem{lemma}[standardlocalcounter]{Lemma}
\newtheorem{proposition}[standardlocalcounter]{Proposition}
\newtheorem{corollary}[standardlocalcounter]{Corollary}
\newtheorem{observation}[standardlocalcounter]{Observation}
\newtheorem{fact}[standardlocalcounter]{Fact}

\newtheorem{conjecturelocalcounter}[standardlocalcounter]{Conjecture}
\newtheorem{conjectureglobalcounter}[standardglobalcounter]{Conjecture}
\newtheorem{conjecture}[standardglobalcounter]{Conjecture}
\newtheorem{openquestion}[standardglobalcounter]{Open Question}
\newtheorem{openproblem}[standardglobalcounter]{Open Problem}
\newtheorem{problem}{Problem}

\newtheorem{property}[standardlocalcounter]{Property}
\newtheorem{definition}[standardlocalcounter]{Definition}
\newtheorem{claim}[standardlocalcounter]{Claim}

\newtheorem{algorithm}[standardlocalcounter]{Algorithm}

\newtheorem{remark}[standardlocalcounter]{Remark}
\newtheorem{example}[standardlocalcounter]{Example}

\renewenvironment{proof}[1][Proof]{\par\trivlist
   \item[\hskip \labelsep{\itshape {#1}.}]\prooffont}
   {\hspace*{0pt plus1fill}\fboxsep2.5pt\fboxrule.5pt\raise3pt\hbox{\fbox{}}\endtrivlist}
 }
    {\ifthenelse{\boolean{detectedLMCS}}
      {

\theoremstyle{plain}    

\newtheorem{theorem}[thm]{Theorem}
\newtheorem{lemma}[thm]{Lemma}
\newtheorem{proposition}[thm]{Proposition}
\newtheorem{corollary}[thm]{Corollary}
\newtheorem{observation}[thm]{Observation}

\newtheorem{conjecture}[thm]{Conjecture}
\newtheorem{problem}[thm]{Problem}

\newtheorem{openquestion}{Open Question}
\newtheorem{openproblem}{Open Problem}

\theoremstyle{definition}

\newtheorem{property}[thm]{Property}
\newtheorem{definition}[thm]{Definition}
\newtheorem{claim}[thm]{Claim}
\newtheorem{remark}[thm]{Remark}
\newtheorem{example}[thm]{Example}

 }
      {\ifthenelse{\boolean{detectedIEEE}}
        {

\newtheorem{standardlocalcounter}{Dummy}[section]
\newtheorem{standardglobalcounter}{Dummy}

\theoremstyle{plain}    

\newtheorem{theorem}[standardglobalcounter]{Theorem}
\newtheorem{lemma}[standardglobalcounter]{Lemma}
\newtheorem{proposition}[standardglobalcounter]{Proposition}
\newtheorem{corollary}[standardglobalcounter]{Corollary}
\newtheorem{observation}[standardglobalcounter]{Observation}
\newtheorem{fact}[standardglobalcounter]{Fact}

\newtheorem{conjecture}[standardglobalcounter]{Conjecture}
\newtheorem{openquestion}{Open Question}
\newtheorem{openproblem}{Open Problem}
\newtheorem{problem}{Problem}

\theoremstyle{definition}

\newtheorem{property}[standardglobalcounter]{Property}
\newtheorem{definition}[standardglobalcounter]{Definition}
\newtheorem{claim}[standardglobalcounter]{Claim}

\theoremstyle{remark}
\newtheorem{remark}[standardglobalcounter]{Remark}
\newtheorem{example}[standardglobalcounter]{Example}

\newtheoremstyle{meta}{3pt}{3pt}{\scshape \small }{}{\scshape \small }{:}{ }{}

\theoremstyle{meta}
\newtheorem{meta}{Meta comment}

\newtheoremstyle{questions}{3pt}{3pt}{\sffamily \slshape}{}{\bfseries \sffamily \slshape}{:}{ }{}

\theoremstyle{questions}
\newtheorem{questions}{Open questions}

 }
        {\ifthenelse{\boolean{detectedLNCS}}
          {

\spnewtheorem*{proofsketch}{Proof sketch}{\itshape}{\rmfamily}

\spnewtheorem{observation}{Observation}{\bfseries}{\itshape}
\spnewtheorem{fact}{Fact}{\bfseries}{\itshape}

 }
          {\ifthenelse{\boolean{detectedACM} \or \boolean{detectedACMconf}}
            {

\theoremstyle{acmplain}
\newtheorem{theorem}{Theorem}[section]        

\newtheorem{observation}[theorem]{Observation}
\newtheorem{fact}[theorem]{Fact}
\newtheorem{claim}[theorem]{Claim}
\newtheorem{property}[theorem]{Property}
\newtheorem{subclaim}[theorem]{Subclaim}
\newtheorem{openquestion}{Open Question}
\newtheorem{openproblem}{Open Problem}

 }
            {\ifthenelse{\boolean{detectedToC}}
              {

\theoremstyle{plain}
\newtheorem{observation}[theorem]{Observation}
\newtheorem{openproblem}[theorem]{Open Problem}

\theoremstyle{definition}
\newtheorem{property}[theorem]{Property}

\renewcommand{\refth}[1]{\expref{Theorem}{#1}}
\renewcommand{\reflem}[1]{\expref{Lemma}{#1}}

\renewcommand{\refdef}[1]{\expref{Definition}{#1}}

\renewcommand{\refobs}[1]{\expref{Observation}{#1}}

\renewcommand{\refsec}[1]{\expref{Section}{#1}}

\renewcommand{\refapp}[1]{\expref{Appendix}{#1}}

\renewcommand{\reffig}[1]{\expref{Figure}{#1}}

 }
              {\ifthenelse{\boolean{detectedLIPIcs}}
                {

\theoremstyle{plain}    

\newtheorem{fact}[theorem]{Fact}
\newtheorem{observation}[theorem]{Observation}
 }
                {\ifthenelse{\boolean{detectedCompCplx}}
                  {

 }
                  {\ifthenelse{\boolean{detectedSigplanconf}}
                    {

\usepackage{amsthm}          
\usepackage{amsthm-boldfont}
\usepackage{url}

\newtheorem{standardlocalcounter}{Dummy}[section]
\newtheorem{standardglobalcounter}{Dummy}

\theoremstyle{plain}    

\newtheorem{theorem}[standardlocalcounter]{Theorem}
\newtheorem{lemma}[standardlocalcounter]{Lemma}
\newtheorem{proposition}[standardlocalcounter]{Proposition}
\newtheorem{corollary}[standardlocalcounter]{Corollary}
\newtheorem{observation}[standardlocalcounter]{Observation}
\newtheorem{fact}[standardlocalcounter]{Fact}

\newtheorem{conjecturelocalcounter}[standardlocalcounter]{Conjecture}
\newtheorem{conjectureglobalcounter}[standardglobalcounter]{Conjecture}
\newtheorem{conjecture}[standardglobalcounter]{Conjecture}
\newtheorem{openquestion}[standardglobalcounter]{Open Question}
\newtheorem{openproblem}[standardglobalcounter]{Open Problem}
\newtheorem{problem}[standardglobalcounter]{Problem}
\newtheorem{question}[standardglobalcounter]{Question}

\theoremstyle{definition}

\newtheorem{property}[standardlocalcounter]{Property}
\newtheorem{definition}[standardlocalcounter]{Definition}
\newtheorem{claim}[standardlocalcounter]{Claim}
\newtheorem{subclaim}[standardlocalcounter]{Subclaim}

\newtheorem{algorithm}[standardlocalcounter]{Algorithm}

\theoremstyle{remark}
\newtheorem{remark}[standardlocalcounter]{Remark}
\newtheorem{example}[standardlocalcounter]{Example}

 }
                    {\ifthenelse{\boolean{detectedAAAI}}
                      {}
                      {\ifthenelse{\boolean{detectedArticle} 
                          \or \boolean{detectedElsevier}
                          \or \boolean{detectedEasyChair}}
                        {

\newtheorem{standardlocalcounter}{Dummy}[section]
\newtheorem{standardglobalcounter}{Dummy}

\theoremstyle{plain}    

\newtheorem{theorem}[standardlocalcounter]{Theorem}
\newtheorem{lemma}[standardlocalcounter]{Lemma}
\newtheorem{proposition}[standardlocalcounter]{Proposition}
\newtheorem{corollary}[standardlocalcounter]{Corollary}
\newtheorem{observation}[standardlocalcounter]{Observation}

\newtheorem{conjecturelocalcounter}[standardlocalcounter]{Conjecture}
\newtheorem{conjectureglobalcounter}[standardglobalcounter]{Conjecture}
\newtheorem{conjecture}[standardglobalcounter]{Conjecture}
\newtheorem{openquestion}[standardglobalcounter]{Open Question}
\newtheorem{openproblem}[standardglobalcounter]{Open Problem}
\newtheorem{problem}[standardglobalcounter]{Problem}

\theoremstyle{definition}

\newtheorem{property}[standardlocalcounter]{Property}
\newtheorem{definition}[standardlocalcounter]{Definition}
\newtheorem{claim}[standardlocalcounter]{Claim}

\theoremstyle{remark}
\newtheorem{remark}[standardlocalcounter]{Remark}
\newtheorem{example}[standardlocalcounter]{Example}

 }
                        {

\newtheorem{standardlocalcounter}{Dummy}[chapter]
\newtheorem{standardglobalcounter}{Dummy}

\theoremstyle{plain}    

\newtheorem{theorem}[standardlocalcounter]{Theorem}
\newtheorem{lemma}[standardlocalcounter]{Lemma}
\newtheorem{proposition}[standardlocalcounter]{Proposition}
\newtheorem{corollary}[standardlocalcounter]{Corollary}
\newtheorem{observation}[standardlocalcounter]{Observation}

\theoremstyle{definition}

\newtheorem{definition}[standardlocalcounter]{Definition}
\newtheorem{claim}[standardlocalcounter]{Claim}

\theoremstyle{remark}

\newtheoremstyle{meta}{3pt}{3pt}{\scshape \small }{}{\scshape \small }{:}{ }{}

\theoremstyle{meta}

\newtheoremstyle{questions}{3pt}{3pt}{\sffamily \slshape}{}{\bfseries \sffamily \slshape}{:}{ }{}

\theoremstyle{questions}

 }}}}}}}}}}}}}

\ifthenelse{\boolean{detectedThesis}}
{}
{\usepackage{ifpdf}}

\ifthenelse{\boolean{detectedToC}}
{}
{\usepackage{graphicx}}

\ifpdf         
\fi

\ifthenelse
{\boolean{detectedSTOC}   \or \boolean{detectedThesis} \or 
 \boolean{detectedLNCS}   \or \boolean{detectedToC}    \or 
 \boolean{detectedLIPIcs} \or \boolean{detectedAAAI}   \or
 \boolean{detectedIJCAI}  \or \boolean{detectedSIAM}}
{}
{
  \setcounter{secnumdepth}{3}
  \setcounter{tocdepth}{3}
}

\newcount\hours
\newcount\minutes
\def\SetTime{\hours=\time
\global\divide\hours by 60
\minutes=\hours
\multiply\minutes by 60
\advance\minutes by-\time
\global\multiply\minutes by-1 }
\SetTime
\def\now{\number\hours:\ifnum\minutes<10 0\fi\number\minutes}

\newcommand{\proofstd}{\pi}

\newcommand{\emptycl}{\bot}

\newcommand{\formf}{\ensuremath{F}}

\ifthenelse{\boolean{detectedACMconf}}
{}
{}

\ifthenelse{\boolean{detectedACMconf}}
{
 }
{
 }

\newcommand{\clb}{\ensuremath{B}}
\newcommand{\clc}{\ensuremath{C}}
\newcommand{\cld}{\ensuremath{D}}

\newcommand{\pcsp}{\polynomialsetformat{P}}

\newcommand{\pcpolyp}{P}
\newcommand{\pcpolyq}{Q}
\newcommand{\pcmonm}{m}

\newcommand{\setsofvarsorlit}[2]{\mathit{#1}({#2})}
\newcommand{\Setsofvarsorlit}[2]{\mathit{#1}\bigl({#2}\bigr)}

\newcommand{\vars}[1]{\setsofvarsorlit{Vars}{#1}}

\newcommand{\Vars}[1]{\Setsofvarsorlit{Vars}{#1}}

\newcommand{\restrict}[2]{{{#1}\!\!\upharpoonright_{#2}}}

\newcommand{\genericformsmall}[2]{\mathit{#1}( #2 )}

\newcommand{\mdegreestd}{d}
\newcommand{\mdegreeof}[2][]{\genericformsmall{Deg_{#1}}{#2}}

\newcommand{\idealarg}[1]{\langle {#1} \rangle}
\newcommand{\idealforterm}[1]{\idealarg{\msupport{#1}}}
\newcommand{\msupport}[1]{S(#1)}

\newcommand{\redideal}[2]{R_{#1}(#2)}

\newcommand{\ropname}{\tilde{R}}
\newcommand{\roperator}[1]{\ropname(#1)}

\newcommand{\razborovoperator}{\ropname}

\newcommand{\monomm}{m}

\newcommand{\erdosrenyi}{Erd\H{o}s-R\'{e}nyi\xspace}

\DeclareMathOperator{\clop}{Cl}

\renewcommand{\epsilon}{\varepsilon}
\newcommand{\gnd}{\mathbb{G}_{n, d}}
\newcommand{\gnparg}[2]{\mathbb{G}({#1}, {#2})}

\newcommand{\gnp}{\gnparg{n}{p}}
\newcommand{\gndn}{\mathbb{G}(n, d/n)}

\renewcommand{\F}{\mathbb{F}}

\newcommand{\multring}{\F[x_1, \ldots, x_n]/\langle x_1^2 - x_1, \ldots, x_n^2 - x_n\rangle}

\newcommand{\graphpropp}{P}
\newcommand{\booleanaxioms}{\{x_1^2-x_1, \ldots, x_n^2- x_\nvar\}}
\newcommand{\polyvars}{x_1, \ldots, x_\nvar}
\newcommand{\Desc}[1]{\mathrm{Desc}( #1 )}
\newcommand{\Col}[2]{\mathrm{Col}(#1, #2)}
\newcommand{\hoplength}{\tau}
\newcommand{\hopexample}{Q}
\newcommand{\Smap}[1]{S( #1 )}

\newcommand{\graphpath}{P}

\newcommand{\polyinu}{q}

\newcommand{\nbhstd}[2]{N_{#1}( #2 )}
\newcommand{\kcolourcons}{k}
\newcommand{\gcolourability}{c}

\newcommand{\indeterminate}{x}

\newcommand{\closure}[2][]{\ensuremath\clop_{{#1}}( #2 )}

\newcommand{\starg}{B}

\newcommand{\eventa}{\mathcal{A}}

\renewcommand{\mdegreestd}{D}
\renewcommand{\pcpolyp}{p}
\renewcommand{\pcpolyq}{q}
\renewcommand{\pcsp}{\mathcal{P}}

\newcommand{\idealS}[1]{\langle S(#1)\rangle}

\newcommand{\sparseparam}{a}

\newcommand{\eps}{\varepsilon}

\newcommand{\badset}{T}

\newcommand{\indic}{\mathbf{1}}

\newcommand{\acceptablename}{support}
\newcommand{\acceptable}[1]{$#1$\nobreakdash-\acceptablename\xspace}

\newcommand{\subsetcalP}{\mathcal{Q}}

\newcommand{\Wset}{W}

\newcommand{\Dset}{Z}

\newcommand{\satcondition}{satisfiability\xspace}
\newcommand{\satlemma}{satisfiability lemma\xspace}
\newcommand{\Satlemma}{Satisfiability lemma\xspace}

\newcommand{\monomialorder}{\prec}
\newcommand{\vertexorder}{\prec_{\mathsf{v}}}

\ifthenelse{\boolean{publisherversion}}
{}
{
  \numberwithin{equation}{section}
}

\usepackage{newpxtext}
\usepackage{newpxmath}
\theoremstyle{theorem}

\newtheorem{innercustomthm}{Lemma}
\newenvironment{restatablelem}[1]
  {\renewcommand\theinnercustomthm{#1}\innercustomthm}
  {\endinnercustomthm} 

\newtheorem{innercustomthmthm}{Theorem}
\newenvironment{restatablethm}[1]
  {\renewcommand\theinnercustomthmthm{#1}\innercustomthmthm}
  {\endinnercustomthmthm}

\begin{document}

\title{Graph
  Colouring Is Hard on Average for \\
  Polynomial Calculus and
  Nullstellensatz\thanks{This is the full-length version of a paper with the same title
    that appeared in the
    \emph{Proceedings of the     
      64th Annual IEEE Symposium on Foundations of Computer Science (FOCS '23)}.}
  }
\author{Jonas Conneryd \\
\textsl{Lund University}\\
\textsl{University of Copenhagen}
\and
Susanna F. de Rezende\\
\textsl{Lund University}
\and
Jakob Nordström \\
\textsl{University of Copenhagen} \\
\textsl{Lund University}
\and
Shuo Pang \\
\textsl{University of Copenhagen} \\
\textsl{Lund University}
\and
Kilian Risse \\
\textsl{EPFL}
}

\date{\today}

\maketitle

\ifthenelse{\boolean{publisherversion}}
{}
{
  \thispagestyle{empty}

  \pagestyle{fancy}
  \fancyhead{}
  \fancyfoot{}
  \renewcommand{\headrulewidth}{0pt}
  \renewcommand{\footrulewidth}{0pt}
 
  \fancyhead[CE]{\slshape 
    GRAPH COLOURING IS HARD ON AVERAGE FOR POLYNOMIAL CALCULUS
    }
  \fancyhead[CO]{\slshape \nouppercase\leftmark}
  \fancyfoot[C]{\thepage}
  
  \setlength{\headheight}{13.6pt}
}

\begin{abstract}

  We prove that polynomial calculus (and hence also Nullstellensatz) over any field requires linear degree to refute that sparse random regular graphs, as well as sparse \erdosrenyi random graphs, are $3$-colourable. Using the known relation between size and degree
  for polynomial calculus proofs,
  this implies
  strongly exponential lower bounds on proof size.
\end{abstract}

\pagenumbering{arabic}
  
\providecommand{\kcolourability}{$k$\nobreakdash-coloura\-bility\xspace}
\providecommand{\kcolouring}{$k$\nobreakdash-colouring\xspace}
\providecommand{\kcolourable}{$k$\nobreakdash-colourable\xspace}

\section{Introduction}
\label{sec:intro}

Determining the \emph{chromatic number} of a graph~$G$, \ie
how many colours are needed for the vertices of~$G$ if no two vertices
connected by an edge should have the same colour,
is one of the
classic  21~problems shown \NP-complete in the seminal work of
Karp~\cite{Karp72Reducibility}.
This
\emph{graph colouring problem},
as it is also referred to, has been
extensively studied
since then,
but there are still major gaps in our understanding.

The currently best known approximation algorithm computes a graph colouring
within
at most a 
factor $\Bigoh{n(\log \log n)^2/(\log n)^3}$ of the chromatic
number~\cite{Halldorsson93StillBetter},
and it is known that approximating this number to within a
factor~$n^{1-\epsilon}$ is
\NP-hard~\cite{Zuckerman07LinearDegreeExtractors}.
Even under the promise that the graph is \mbox{$3$-colourable},
the most 
parsimonious 
algorithm with guaranteed polynomial running
time 
needs
$\Bigoh{n^{0.19996}}$ colours~\cite{KT17Coloring3Colorable}.
This is very far from the lower bounds that are known---it is
\NP-hard to 
\mbox{$(2k-1)$-colour} a \mbox{$k$-colourable} graph~\cite{BBKO21AlgebraicCSP},
but the question of whether colouring a
\mbox{$3$-colourable} graph with $6$~colours is \NP-hard remains
open~\cite{KO22Invitation}.
It is widely believed that any algorithm
that colours graphs optimally
has to run
in exponential time in the worst case, and the currently
fastest algorithm for $3$\nobreakdash-colouring
has time complexity $\Bigoh{1.3289^{n}}$
\cite{Beigel05ColoringFixed}.
A survey on various algorithms and techniques for so-called exact
algorithms for graph colouring can be found in~\cite{Husfeldt15Colouring}.

Graph colouring instances of practical interest might not exhibit such
exponential-time
behaviour, however, and in such a context it is relevant to
study algorithms without worst-case guarantees and examine how they
perform in practice.
To understand such algorithms from 
a computational complexity viewpoint,
it is natural to investigate bounded models of computation
that are strong enough to describe the reasoning 
performed
by the 
algorithms and to prove unconditional lower bounds
that hold in 
these models.

\subsection{Previous Work}
\label{sec:intro-previous}

Focusing on random graphs, 
McDiarmid~\cite{McDiarmid84Colouring}
developed a method for
determining
\kcolourability
that captures a range of algorithmic approaches.
Beame \etal~\cite{BCCM05RandomGraph} showed that
this method could in turn be simulated by the
resolution proof system~\cite{Blake37Thesis,DP60ComputingProcedure,DLL62MachineProgram,Robinson65Machine-oriented},
and
established
average-case exponential lower bounds for resolution proofs of
non-$k$-colourability for random graph instances sampled so as not to be
\kcolourable with exceedingly high probability.

Different algebraic approaches for \kcolourability
have been considered in
\cite{AT92Colorings,Lovasz94Stable,Matiyasevich74Criterion,Matiyasevich04Algebraic}.
Bayer~\cite{Bayer82Division} seems to have been the first to use Hilbert's Nullstellensatz to attack graph colouring.
Informally, the idea is to 
write the problem as a set of polynomial equations
$\setdescr{\pcpolyp_i(x_1,\ldots, x_n) = 0}{i \in [m]}$
in such a way that legal 
$k$-colourings 
correspond to
common roots of these polynomials.
Finding polynomials   $\pcpolyq_1, \ldots, \pcpolyq_m$ 
such that
$\sum_{i=1}^{m} \pcpolyq_i \pcpolyp_i = 1$
then      proves that the graph is not $k$-colourable.
This latter equality is referred to as a 
\introduceterm{Nullstellensatz certificate}
of non-colourability, and the \introduceterm{degree} of this
certificate is the largest degree of any polynomial  $\pcpolyq_i \pcpolyp_i$ in the sum.
Later papers based on Nullstellensatz and Gröbner bases, such as
\cite{DeLoera95Grobner,Mnuk01Representing,HW08Algebraic},
culminated in an award-winning sequence of works
\cite{DLMM08Hilbert,DLMO09ExpressingCombinatorial,DLMM11ComputingInfeasibility,DMPRRSSS15GraphColouring}
presenting algorithms with surprisingly good practical performance.

For quite some time, no strong lower bounds 
were known for these
algebraic methods
or the corresponding proof systems
\introduceterm{Nullstellensatz}~\cite{BIKPP94LowerBounds}
and
\introduceterm{polynomial calculus}
\cite{CEI96Groebner,ABRW02SpaceComplexity}.
On the contrary, the authors
of~\cite{DLMO09ExpressingCombinatorial} reported that essentially 
all benchmarks they studied turned out to have
Nullstellensatz certificates of
small constant degree.
The degree lower bound~$k+1$ for $k$~colours
in~\cite{DMPRRSSS15GraphColouring} remained the best known until
optimal, linear, degree lower bounds for polynomial calculus were
established in~\cite{LN17GraphColouring} using a reduction from
so-called functional pigeonhole principle
formulas~\cite{MN15GeneralizedMethodDegree}.  A more general reduction
framework was devised in~\cite{AO19ProofCplx} to obtain optimal degree
lower bounds also for the proof systems
\mbox{\emph{Sherali-Adams}~\cite{SA90Hierarchy}}
and
\emph{sums-of-squares}~\cite{Lasserre01Explicit,Parrilo00Thesis},
as well as weakly exponential size lower bounds for
\emph{Frege proofs}~\cite{CR79Relative,Reckhow75Thesis}
of bounded depth.

The lower bounds discussed in the previous paragraph are not quite
satisfactory,
in that it is not clear how much they actually tell us about the graph colouring
problem, as opposed to the hardness of the problems being reduced from.
In order to improve our understanding for a wider range of graph instances, 
it seems both natural and desirable to establish average-case
lower bounds for random graphs, just as for resolution
in~\cite{BCCM05RandomGraph}.
However,  this goal
has remained elusive for
almost two decades, as pointed out, e.g., in
\cite{MN15GeneralizedMethodDegree,
  LN17GraphColouring,Lauria2018,BN21ProofCplxSAT}. For sparse random
graphs, where the number of edges is linear in the number of vertices,
no superconstant degree lower bounds at all have been established for
algebraic or semialgebraic proof systems. On the contrary, it was
shown in~\cite{JKM19Lovasz}, improving on~\cite{CojaOghlan05Lovasz},
that degree\nobreakdash-$2$ sums\nobreakdash-of\nobreakdash-squares
refutes $k$\nobreakdash-colourability on random
$d$\nobreakdash-regular graphs \aas whenever ${d \geq 4k^2}$. 
For dense random graphs, the strongest
lower bound seems to be the recent logarithmic
degree
bound in the sums-of-squares proof system for Erd\H{o}s-R\'{e}nyi random graphs with
edge probability~$1/2$ and $k = n^{1/2 + \epsilon}$
colours~\cite{KM21StressFree}.
Since this result is for a problem encoding using inequalities, however,
it is not clear whether this has any implications for Nullstellensatz
or polynomial calculus over the reals (which are known to be polynomially simulated by sums-of-squares). And for other fields
nothing has been known for the latter two proof systems---not even
logarithmic lower bounds.

\subsection{Our Contribution}
\label{sec:intro-contrib}

In this work, we establish optimal,
linear, degree lower bounds
and exponential size lower bounds
for
polynomial calculus proofs of non-colourability of random graphs.

\begin{theorem}[informal]
  \label{th:main-theorem}
  For any $d\geq 6$,
  polynomial calculus (and hence also Nullstellensatz) requires \aas
  linear degree to refute that
  random $d$-regular graphs,
  as well as
  \erdosrenyi random graphs,
  are $3$-colourable. These degree lower bounds hold over any field,
  and also imply exponential lower bounds on proof size.
\end{theorem}

We prove our lower bound for the standard encoding in proof
complexity, where 
binary
variables~$x_{v,i}$ indicate whether vertex~$v$ is
coloured with colour~$i$ or not.
It should be pointed out
that, just as the results
in~\cite{LN17GraphColouring}, 
our degree lower bounds 
also 
apply to the
\kcolourability
encoding
introduced in~\cite{Bayer82Division}
and used in computational algebra papers such as
\cite{DLMM08Hilbert,DLMO09ExpressingCombinatorial,DLMM11ComputingInfeasibility,DMPRRSSS15GraphColouring},
where a primitive $k$th root of unity is adjoined to the field
and different colours of a vertex~$v$
are encoded by a variable~$x_v$
taking different powers of this root of unity.

Our 
lower bound 
proofs crucially use
a new idea for proving degree lower bounds for colouring graphs with
large girth~\cite{RT22GraphsLargeGirth}.
After adapting this approach
from the root-of-unity encoding to the
Boolean indicator variable encoding, and replacing the proof in
terms of girth with
a strengthened argument using
carefully chosen properties of random graphs, we
obtain a
remarkably
clean and simple solution to the long-standing
open problem of showing average-case polynomial calculus degree lower
bounds for graph colouring.
We elaborate on our techniques in more detail next.

\subsection{Discussion of Proof Techniques}
\label{sec:intro-techniques}

In most works on algebraic and semialgebraic proof systems such as
Nullstellensatz, polynomial calculus,
Sherali-Adams, and sums-of-squares,
the focus has been on proving upper and lower bounds on the
degree
of
proofs. Even when proof size is the measure of interest, almost all
size lower bounds have been established via degree lower bounds
combined with general results saying that for all of the above proof
systems except Nullstellensatz strong enough lower bounds on degree
imply lower bounds on size~\cite{IPS99LowerBounds,AH19SizeDegree}.

At a high level, the techniques for proving degree lower bounds for the
different proof systems have a
fairly   
similar flavour.
For the static proof systems, \ie Nullstellensatz, Sherali-Adams, and
sums-of-squares, 
it is enough to
show that the dual
problem   is feasible and thus rule
out
low-degree
proofs.
In more detail, for Nullstellensatz, one constructs a
\emph{design}~\cite{Buss98LowerBoundsNS}, which is a linear functional
mapping low-degree monomials to
elements in   the underlying field. This functional should map low-degree monomials multiplied
by any input polynomial~$\pcpolyp_i$ 
to~$0$, but should map $1$ to a non-zero field element.
If such a functional can be found, it is clear that there cannot exist
any low-degree Nullstellensatz certificate
$\sum_{i=1}^{m} \pcpolyq_i \pcpolyp_i = 1$ of unsatisfiability, as the design would map the left-hand side of the equation to zero
but the right-hand side to non-zero.
For Sherali-Adams, the analogous functional furthermore has to map
any low-degree monomials to non-negative numbers,
and for sums-of-squares this should also hold for
squares of low-degree polynomials.
Such a
\emph{pseudo-expectation} can be viewed as
an expectation over a
fake
probability
distribution
supported on
satisfying assignments to the
problem, which
is indistinguishable
from a true distribution for an adversary using only low-degree
polynomials.

Polynomial calculus is different from
these proof systems
in that
it does not present the certificate of unsatisfiability as a static
object, but instead,
given a set of polynomials $\mathcal{P}$,
dynamically derives new polynomials in the ideal
generated by~$\mathcal{P}$.
The derivation ends
when it reaches the polynomial~$1$, \ie the multiplicative identity in
the field, showing that there is no solution.
The most common way to prove degree lower bounds is to design a
\emph{pseudo-reduction operator}~\cite{Razborov98LowerBound}, which maps all
low-degree polynomials derived from $\mathcal{P}$ to~$0$ but sends $1$
to~$1$,  and which is indistinguishable from a true ideal reduction
operator if one is limited to reasoning with low-degree
polynomials. This means that for bounded-degree polynomial calculus derivations
it seems like
the set of input polynomials are consistent.

Following the method in~\cite{AR03LowerBounds},
a pseudo-reduction operator~$\ropname$ can be constructed by 
defining it on low-degree monomials and extending
it
to
polynomials by linearity.
For every monomial~$\monomm$, we
identify 
a set of related input
polynomials~$\msupport{\monomm}$,
let $\idealforterm{\monomm}$ be the ideal generated by these
polynomials,
and define
$
\roperator{\monomm}
=
\redideal{\idealforterm{\monomm}}{\monomm}
$
to be the reduction of~$\monomm$ modulo the ideal~$\idealforterm{\monomm}$.
Intuitively, we think of $\msupport{\monomm}$ as the 
(satisfiable) subset of polynomials that
might possibly have been used in a low-degree derivation of~$\monomm$,
but since the constant monomial~$1$ is not derivable in low degree it gets an
empty associated set of polynomials, 
meaning that
$
\roperator{1}
=
\redideal{\idealforterm{1}}{1}
=
1
$.
In order for~$\ropname$ to look like a real reduction operator,
we need to show 
that 
for polynomials $\pcpolyp$ and~$\pcpolyp'$
of not too high degree
it holds that   $
\roperator{\pcpolyp + \pcpolyp'}
=
\roperator{\pcpolyp} +
\roperator{\pcpolyp'}
$
and
$
\roperator{\pcpolyp \cdot \roperator{\pcpolyp'}}
=
\roperator{\pcpolyp \cdot \pcpolyp'}
$.
The first equality is immediate, since $\ropname$ is defined to be a linear operator,
but the second equality
is more problematic.
Since the polynomials~$\pcpolyp$ and~$\pcpolyp'$ will be reduced
modulo different ideals---in fact, this will be the case even for
different monomials within the same polynomial---a priori there is no
reason why $\ropname$~should
commute with multiplication.

Proving that a pseudo-reduction operator $\razborovoperator$ behaves
like
an actual
reduction operator
for low-degree polynomials is typically the most challenging technical
step in the lower bound proof.
Very roughly, the proof method
in~\cite{AR03LowerBounds} goes as follows. Suppose that 
$\monomm$ and~$\monomm'$ are monomials with associated polynomial sets
$\msupport{\monomm}$
and~$\msupport{\monomm'}$, respectively.
Using expansion properties
of the constraint-variable incidence graph for the
input polynomials,
we argue that the true reduction operator will not change if
we reduce both monomials modulo the larger ideal
$
\idealarg{\msupport{\monomm} \union
  \msupport{\monomm'}}
$
generated by the union of their associated sets of polynomials.
This implies that we have
$
\roperator{\monomm'}
=
\redideal{\idealforterm{\monomm'}}{\monomm'}
=
\redideal{\idealarg{
    \msupport{\monomm} \union
    \msupport{\monomm'} 
  }}{\monomm'}
$
and
$
\roperator{\monomm \cdot \monomm'}
=
\redideal{\idealarg{
    \msupport{\monomm} \union
    \msupport{\monomm'} 
  }}{\monomm \cdot \monomm'}
$,
from which it follows that
$
\roperator{\monomm \cdot \roperator{\monomm'}}
=
\roperator{\monomm \cdot \monomm'}
$
holds, just like for reduction modulo an actual ideal.
To prove that expanding the ideals does not change the reduction operator
is a delicate balancing act, though, since the
ideals will need to be large enough to guarantee non-trivial reduction,
but at the same time small enough so that different ideals can be
``patched together'' with only local adjustments.

All previous attempts to apply this lower bound strategy to the graph
colouring problem have failed. For other polynomial calculus lower
bounds it has been possible to limit the interaction between different
polynomials in the input. For graph colouring, however, applying the
reduction operator intuitively corresponds to partial colourings of
subsets of vertices, and it has not been known how to avoid that
locally assigned colours propagate new colouring constraints through
the rest of the graph.  In technical language, what is needed is a way
to order the vertices in the graph so that there will be no long
ordered paths of vertices along which colouring constraints can
spread.  It has seemed far from obvious how to construct such an
ordering, or even whether such an ordering should exist, and due to
this technical problem it has not been possible to join local ideal
reduction operators into a globally consistent pseudo-reduction
operator.

This
technical
problem was addressed in a recent
paper~\cite{RT22GraphsLargeGirth} by an ingenious, and in hindsight
surprisingly simple, idea.  The main insight is to consider a proper colouring
of the graph~$G$ with $\chi(G)$ colours,
and then order the vertices in each colour class consecutively. In
this way, order-decreasing paths are of length at most~$\chi(G)$,
and one can guarantee some form of locality.
Once this order is in place, the final challenge is to 
ensure that small cycles do not 
interfere when ``patching together'' reductions. 
In~\cite{RT22GraphsLargeGirth},
such conflicts are avoided
by ensuring that the graph should have high girth,
which
results
in a degree lower bound linear in the girth of the
graph. 
In terms of graph size, this cannot give better than logarithmic lower
bounds, however, 
since the girth is at most logarithmic in the number of vertices for
any graph with chromatic number larger \mbox{than $3$~\cite{Bollobas87Girth}.}

In our work, we
employ the same ordering as in~\cite{RT22GraphsLargeGirth}, but
instead of girth use the fact that random graphs are locally very
sparse. Once the necessary technical concepts are in place, the proof
becomes quite simple and elegant, which we view as an additional
strength of our result.

\subsection{Outline of This Paper}
The rest of this paper is organized as follows. In
\refsec{sec:prelims} we present some preliminaries and then, as a
warm-up, reprove the resolution width lower bound for the
colourability formula in \refsec{sec:resolution-lb}. After this, we
proceed to revisit the general framework to obtain polynomial calculus
lower bounds in \refsec{sec:techniques}. In \refsec{sec:closure} we
introduce the important notion of a closure and in
\refsec{sec:3-colouring-hard-merged} prove our main theorem. We
conclude with some final remarks and open problems in
\refsec{sec:conclusion}.

\section{Preliminaries}
\label{sec:prelims}

Let us start by briefly reviewing the necessary preliminaries from
proof complexity, graph theory, and algebra.  We use standard
asymptotic notation. In this paper $\log$ denotes the logarithm base
$2$, while $\ln$ denotes the natural logarithm.

For a field $\F$ we let $\F[x_1, \ldots, x_\nvars]$ denote the
polynomial ring over $\F$ in $\nvars$ variables and let a
\emph{monomial} denote a product of variables. We denote by $\vars{m}$
the \emph{variables of a monomial} $m$, that is, if
$m = \prod_{i\in I} x_i$, then $\vars{m} = \cup_{i\in I} x_i$ and
extend this notation to polynomials $p = \sum_m a_m m$ by
$\vars{p} = \cup_m \vars{m}$. For polynomials
$p_1, \ldots, p_m \in \F[x_1, \ldots, x_n]$ we let
$\langle p_1, \ldots, p_m \rangle$ denote the ideal generated by these
polynomials: $\langle p_1, \ldots, p_m \rangle$ contains all
polynomials of the form $\sum_{i =1}^m q_i p_i$, for
$q_i \in \F[x_1, \ldots, x_n]$.
For a polynomial~$\polyinu$ and a partial function $\rho$ mapping
variables to polynomials we let $\restrict{\polyinu}{\rho}$ denote the
polynomial obtained from $\polyinu$ by substituting every occurrence
of a variable $x_i$ in the domain of $\rho$ by $\rho(x_i)$.

\subsection{Proof Complexity} 
\label{sec:proof-complexity}
\emph{Polynomial calculus (PC)}~\cite{CEI96Groebner} is a proof system
that uses algebraic reasoning to deduce that a system $\pcsp$ of
polynomials over a field $\F$ involving the variables $\polyvars$ is
infeasible, \ie that the polynomials in~$\pcsp$ have no common
root. Polynomial calculus interprets $\pcsp$ as a set of generators of
an ideal and derives new polynomials in this ideal through two
derivation rules:
\begin{subequations}
\begin{align}
    \label{eq:pc-lin-combn}
  \textit{Linear combination: }
  &\frac{\pcpolyp \quad \pcpolyq}{a\pcpolyp + b\pcpolyq},\; a, b \in \F \, ;\\
    \label{eq:pc-mult}
  \textit{Multiplication: }
  & \frac{\pcpolyp}{\indeterminate_i \pcpolyp},\; \indeterminate_i \text{ any variable.}
\end{align}
\end{subequations}
A \emph{polynomial calculus derivation} $\proofstd$ of a
polynomial~$\pcpolyp$ starting from the set~$\pcsp$
is a sequence
of \mbox{polynomials  $(\pcpolyp_1, \ldots, \pcpolyp_\tau)$}, where
$\pcpolyp_\tau = \pcpolyp$ and each polynomial $\pcpolyp_i$
either is
in~$\pcsp$ or
is   obtained by applying one of the derivation rules
\refeq{eq:pc-lin-combn}-\refeq{eq:pc-mult}
to polynomials $\pcpolyp_j$ with $j < i$. A \emph{polynomial calculus
  refutation of $\pcsp$} is a derivation of the constant polynomial
$1$ from $\pcsp$. 
It is well-known 
that polynomial
calculus is sound and complete
when the system $\pcsp$ contains all Boolean axioms $\booleanaxioms$.
We often refer to $\pcsp$ as the set of \emph{axioms},
and we say that a subset of axioms $\subsetcalP\subseteq\pcsp$
is \emph{satisfiable} if the polynomials in $\subsetcalP$ have a 
common root.

The most common complexity measures
of polynomial calculus refutations are \emph{size} and \emph{degree}.
The \emph{size} of a polynomial $p$ is its number of monomials when
expanded into a linear combination of distinct monomials, and the
\emph{degree} of $\pcpolyp$ is the maximum degree among all
of its monomials. The size of a polynomial calculus refutation
$\proofstd$ is the sum of the sizes of the polynomials in~$\proofstd$,
and the degree of $\proofstd$ is the maximum degree among all
polynomials in~$\proofstd$. We follow the convention of not counting
applications of the Boolean axioms toward degree or size by tacitly
working over ${\multring}$, which only strengthens a lower bound on
either measure.
Polynomial calculus size and degree are connected through the
\emph{size-degree relation}~\cite{IPS99LowerBounds}: if $\pcsp$
consists of polynomials with initial degree $d$ and 
contains all Boolean axioms,
and if $\mdegreestd$ is
the minimal degree among all polynomial calculus refutations of
$\pcsp$, then every refutation of $\pcsp$ must have size
$\exp\big(\Bigomega{(D - d)^2/\nvars}\big)$.

The size-degree relation also applies to the stronger proof system
\emph{polynomial calculus resolution (PCR)}
\cite{ABRW02SpaceComplexity}, which is polynomial calculus where
additionally each variable $x_i$ appearing in~$\pcsp$ has a formal
negation $\overline{x}_i$, enforced by adding polynomials
$x_i + \overline{x}_i - 1$ to $\pcsp$. Polynomial calculus and PCR are
equivalent \wrt degree, since the map $\overline{x}_i \mapsto 1-x_i$
sends any PCR proof to a valid polynomial calculus proof of the same
degree. Therefore, to prove a lower bound on PCR size it suffices to
prove a lower bound on polynomial calculus degree, and in particular
all size lower bounds in this paper also apply to PCR. Finally, we
remark that lower bounds on polynomial calculus degree or size also
apply to the weaker \emph{Nullstellensatz} proof system mentioned in
\refsec{sec:intro-previous} and \refsec{sec:intro-contrib}.

\subsection{Graph Colouring and Pseudo-Reductions}

Given a graph $G$, we study the polynomial calculus degree required to
refute the system $\Col{G}{k}$ of polynomials
\begin{subequations}
  \begin{align}
    \label{eq:binary-encoding1}
    \sum_{i=1}^\kcolourcons x_{v, i} - 1
    & \qquad v \in V(G)
    &&[\text{every vertex is assigned a colour}]\\
    \label{eq:binary-encoding2}
    x_{v, i}x_{v, i'}
    & \qquad v \in V(G),\ i \neq i' \in [k]
    &&[\text{no vertex gets more than one colour}] \\
    \label{eq:binary-encoding3}
    x_{u, i}x_{v, i}
    & \qquad (u, v) \in E(G), \ i \in [k]
    &&[\text{no two adjacent vertices get the same colour}] \\
    \label{eq:binary-encoding4}
    x_{v, i}^2 - x_{v, i}
    & \qquad v \in V(G), \ i \in [k]
    &&[\text{Boolean axioms}]
  \end{align}   
\end{subequations}
whose common roots correspond precisely to proper
$k$\nobreakdash-colourings of $G$. We refer to axioms in
\refeq{eq:binary-encoding1} and \refeq{eq:binary-encoding2} as
\emph{vertex axioms} and to \refeq{eq:binary-encoding3} as \emph{edge
  axioms}. Note that the system $\Col{G}{k}$ of polynomials is not the
standard polynomial translation of the usual CNF formula (defined
in~\refsec{sec:resolution}) as \refeq{eq:binary-encoding1} does not
correspond to a single clause. We work with the above formulation for
the sake of exposition and our lower bound also applies to the
standard translation of the CNF formula.

If the field $\F$ we are working over contains, or can be extended to
contain, a primitive $k$th root of unity, then it is
known~\cite[Proposition 2.2]{LN17GraphColouring} that a polynomial
calculus degree lower bound for $\Col{G}{k}$ also applies to
\emph{Bayer's formulation}~\cite{Bayer82Division}
of~$k$\nobreakdash-colourability, where each colour corresponds to a
$k$th root of unity. This encoding has received considerable attention
in computational algebra~\cite{DLMM08Hilbert,
  DLMO09ExpressingCombinatorial,DLMM11ComputingInfeasibility,
  DMPRRSSS15GraphColouring,RT22GraphsLargeGirth}.

Our proof of \refthm{th:main-theorem} is based on the notion of a
\emph{pseudo-reduction} operator or \emph{$R$-operator}. The following
lemma is due to Razborov~\cite{Razborov98LowerBound}. 

\begin{definition}[Pseudo-reduction]
  \label{def:r-op}
  Let $D \in \N^+$ and $\pcsp$ be a set of polynomials over $\F[x_1, \ldots, x_\nvars]$. An $\F$-linear
  operator
  $\razborovoperator\colon \F[x_1, \ldots, x_n] \to \F[x_1, \ldots,
  x_n]$ 
  is a \emph{degree-$D$ pseudo-reduction for $\pcsp$} if
  \begin{enumerate}
  \item $\razborovoperator(1) = 1$, \label{eq:rprop1}\label{item:1}
    
  \item $\razborovoperator(\pcpolyp) = 0$ for every polynomial
    $\pcpolyp\in\pcsp$, and\label{eq:rprop2}
    \label{item:2}

  \item
    $\razborovoperator(\indeterminate_im) =
    \razborovoperator\bigl(\indeterminate_i
    \razborovoperator(m)\bigr)$ for any monomial $m$ of degree at most
    $\mdegreestd-1$ and any variable $\indeterminate_i$.
    \label{item:3}
  \end{enumerate}
\end{definition}

\begin{lemma}[\cite{Razborov98LowerBound}]
  \label{lem:r-operator}
  Let $D \in \N^+$ and $\pcsp$ be a set of polynomials over
  $\F[x_1, \ldots, x_\nvars]$. If there is a degree-$D$
  pseudo-reduction for $\pcsp$, then any polynomial calculus
  refutation of $\pcsp$ over $\F$ requires degree strictly greater
  than $D$.
\end{lemma}

The proof of \reflem{lem:r-operator} is straightforward: apply
$\razborovoperator$ to all polynomials in a purported polynomial
calculus refutation of~$\pcsp$ and conclude by induction that it is
impossible to reach contradiction in degree at most~$D$.

\subsection{Algebra Background}
\label{sec:algebra}

The definition of our pseudo-reduction operator requires some standard
notions from algebra.
A total well-order $\prec$ on the
monomials in~$\F[x_1,\ldots,x_n]$ is \emph{admissible} if
the following properties hold:
\begin{enumerate}
\item  For any monomial $m$ it holds that $1 \prec m$. 
\item For any monomials $m_1, m_2$ and $m$ such that~${m_1 \prec m_2}$, it holds
  that~$mm_1 \prec mm_2$.
\end{enumerate}
Note that this definition is the more standard definition of {admissible}
order (also known as monomial order) used in algebra, and differs from that
introduced in~\cite{Razborov98LowerBound}, and used subsequently 
in~\cite{AR01LowerBounds,MN15GeneralizedMethodDegree},
since it is defined over any monomial 
in~$\F[x_1,\ldots,x_n]$ and not only multilinear monomials,
and it does not necessarily have to respect the total degree
of monomials.
We use the above definition due to its simplicity. Recall that we only
concern ourselves with multilinear polynomials in this article by tacitly
working over ${\multring}$.

We write $m_1\preceq m_2$ to denote that $m_1 \prec m_2$ or
$m_1 = m_2$. The \emph{leading monomial} of a polynomial~$p$ is the
largest monomial according to~$\prec$ appearing in~$p$ with a non-zero
coefficient. For an ideal~$I$ over~$\F[x_1,\ldots,x_n]$, a monomial~$m$ is \emph{reducible modulo $I$} if $m$ is the leading
monomial of some polynomial~$\pcpolyq\in I$; otherwise $m$ is
\emph{irreducible modulo $I$}. Under a total monomial order it is
well\nobreakdash-known that for any ideal $I$ and any
polynomial~$\pcpolyp$ there exists a unique representation
$\pcpolyp = \pcpolyq+r$ such that $\pcpolyq \in I$ and $r$ is a linear
combination of irreducible monomials modulo $I$. We call~$r$ the
\emph{reduction of $\pcpolyp$ modulo $I$}, and denote by $R_I$ the
\emph{reduction operator} which maps polynomials~$p$ to the reduction
of $p$ modulo $I$. It is straightforward to verify that the reduction
operator is linear over the vector space of polynomials.

We will use the following standard fact.

\begin{observation}
  \label{obs:ideal-contain-reduction}
  If $I_1$ and $I_2$ are ideals over $\F[x_1, \ldots, x_n]$ and
  $I_1 \subseteq I_2$, then for any monomials $m$ and $m'$ it holds
  that~${R_{I_2}\bigl(m'R_{I_1}(m)\bigr) = R_{I_2}(m'm)}$.
\end{observation}
\begin{proof}
  Write $m = q_1 + r_1$ where $r_1$ is the reduction of $m$ modulo
  $I_1$ and hence $q_1 \in I_1$. Similarly, let
  $m'\cdot R_{I_1}(m) = m' \cdot r_1 = q_2 + r_2$ for $q_2 \in I_2$ and $r_2$ the
  reduction of $m' \cdot r_1$ modulo $I_2$. We have that
  $R_{I_2}\bigl(m'R_{I_1}(m)\bigr) = r_2$ and we now argue that 
  $R_{I_2}(m'm) = r_2$.  Note that $m'(m-q_1) = r_2 + q_2$
  and hence $m'\cdot m = r_2 + q_2 + m' \cdot q_1$. Since $I_1 \subseteq I_2$ and
  ideals are closed under multiplication and addition it holds that
  $q_2 + m_2q_1 \in I_2$. Moreover, since $r_2$ is irreducible modulo $I_2$,
  it follows that $R_{I_2}(m'm) = r_2$ by uniqueness of a
  reduction modulo an ideal. \end{proof}

Finally, let us record the following form of Hilbert's Nullstellensatz
on the Boolean cube.

\begin{lemma}
  \label{lem:implication-ideal}
  Let $g$ be a polynomial and $Q$ be a set of polynomials
  in~$\F[x_1, \ldots, x_n]$, and suppose that $Q$ contains all the
  Boolean axioms. Then it holds that $g$ vanishes on all common roots
  of~$Q$ if and only if $g \in \langle Q \rangle$.
\end{lemma}

Note that the interesting direction of \reflem{lem:implication-ideal}
immediately follows from the implicational completeness of the
Nullstellensatz proof system. For the convenience of the reader we
provide a self-contained proof in \refapp{sec:polyclaim}.

\subsection{Graph Theory}
\label{sec:graph-theory}
We consider graphs $G=(V,E)$ that are finite and undirected, and
contain no self-loops or multi-edges.
Given a vertex set $U \subseteq V$, the \emph{neighbourhood of $U$} in
$G$ is~${N(U) = \{v \in V \mid \exists u \in U\colon (u, v) \in E\}}$,
and for a second set~$W \subseteq V$ we let the \emph{neighbourhood of
  $U$ in~$W$} be denoted by $\nbhstd{W}{U} = N(U) \intersection W$.
The set of edges between vertices in~$U$ is denoted by~$E(U)$ and 
we let $G[U]$ denote the subgraph induced by $U$ in $G$, that is,
$G[U] = (U, E(U))$. A graph is said to be \emph{$d$-regular} if all
vertices are of degree $d$. Note that a graph on $n$ vertices can be
$d$-regular only if $d < n$ and $dn$ is even.

For an edge $e = (u, v)\in E$ the \emph{contraction} of $G$ \wrt $e$
is the graph obtained from $G$ by identifying the vertices $u$ and $v$
as a single, new, vertex $v_e$ and adding an edge between a vertex
$w \in V \setminus \set{u,v}$ and $v_e$ if and only if there is an
edge between $w$ and at least one of $u$ or $v$ in $G$. For a set of
edges $S \subseteq E$ we let the \emph{contraction of $G$ \wrt $S$} be
the graph obtained from $G$ by contracting the edges in $S$ one at a
time, in any order.

A graph is said to be \emph{$k$-colourable} if there is a
mapping~$\chi\colon V \to [k]$ satisfying~$\chi(u) \neq \chi(v)$ for
all edges ${(u, v)\in E}$, in which case we refer to $\chi$ as a
\emph{proper $k$-colouring of $G$}. The \emph{chromatic number} of
$G$, denoted by $\chi(G)$, is the smallest integer $k$ such that $G$
is \emph{$\kcolourcons$-colourable}.

\begin{definition}[Sparsity]\label{def:sparse}
  A graph $G=(V,E)$ is
  \mbox{\emph{$(\ell, \epsilon)$\nobreakdash-sparse}} if every vertex
  set~$U\subseteq V$ of size at most $\ell$ satisfies
  $\setsize{E(U)} \leq (1+ \epsilon)\setsize{ U }$.
\end{definition}

The essential property of sparse graphs we use to obtain our main
result is that large subsets of vertices in such graphs are
$3$-colourable.

\begin{lemma}
  \label{lem:sparse-colourable}
  If a graph $G=(V, E)$ is $(\ell, \epsilon)$\nobreakdash-sparse for
  some $\epsilon < 1/2$, then it holds for every subset $U\subseteq V$
  of size at most $\ell$ that~$G[U]$ is $3$-colourable.
\end{lemma}
  
\begin{proof}
  By induction on $\setsize{U}$. The base case $\setsize{U} = 1$ is
  immediate. For the inductive step we may assume that the claim holds
  for sets of size at most $s-1$. Consider a set~$U\subseteq V$ of
  size $s \leq \ell$. The average degree of a vertex in~$G[U]$
  is~$2\setsize{E(U)}/s$, which is at most $2(1+\epsilon) < 3$ by the
  assumption on sparsity. Hence, since graph degrees are integers,
  there exists a vertex~$v\in U$ with degree at most 2 in~$G[U]$. The
  graph~$G{[U\setminus \set{v}]}$ is $3$-colourable by the inductive
  hypothesis, and every $3$-colouring witnessing this will leave at
  least one colour available to properly colour $v$. Hence every
  $3$-colouring of~${G[U\setminus \set{v}]}$ can be extended to
  $G[U]$, which concludes the proof.
\end{proof}

We consider two models of random graphs. One is the
\emph{Erd\H{o}s-R\'enyi random graph model $\gnp$}, which is the
distribution over graphs on $n$ vertices where each edge is included
independently with probability $p$. The other is the \emph{random
  $d$-regular graph model~$\gnd$}, which is the uniform distribution
over $d$-regular graphs on $n$ vertices. A graph property $P$ holds
\emph{\aas} for a random graph model
$\mathbb{G} = \set{\mathbb{G}_n}_{n=1}^\infty$ if~$\lim_{n \to \infty} \Pr_{G\sim\mathbb{G}_n} [\text{$G$ has property
  $\graphpropp$}] = 1$.

Random graphs are sparse with excellent parameters, as stated in the
following lemma, which is essentially due
to~\cite{Razborov17WidthSemialgebraic}. (See \refapp{sec:sparsity} for
a proof.)

\begin{lemma}[Sparsity lemma]
\label{lem:razborov-sparse}
  For $n,d\in \N^+$ and $\epsilon, \delta \in \R^+$ such that $\epsilon \delta = \omega(1/\log n)$, the following holds  \aas.
  \begin{enumerate}
  \item If $G$ is a graph sampled from~$\gndn$, then it is
    $((4d)^{- (1+\delta)(1+\epsilon)/\epsilon}n, \epsilon)$-sparse.
  \item For $d^2\le \epsilon \delta \log n$, if $G$ 
    is a graph sampled from~$\gnd$, then it is
      $((8d)^{- (1+\delta)(1+\epsilon)/\epsilon}n, \epsilon)$-sparse.
  \end{enumerate}
\end{lemma}

Finally, we need some bounds on the chromatic number of graphs sampled
from $\gndn$ or $\gnd$, where, in particular, the lower bounds ensure
that~$\Col{G}{k}$ is unsatisfiable for large enough~$d$.

\begin{lemma}[\cite{KPGW10Chromatic,CFRR02RandomRegular,AN05Chromatic,Luczak91Chromatic}]
  \label{lem:chromatic-number-random}
  For $n \in \N$, $d \leq n^{0.1}$ and a graph $G$ sampled from either
  $\gndn$ or $\gnd$ it holds \aas that the chromatic number $\chi(G)$
  is at most $2d/\log{d}$ and, if $d \geq 6$, then $\chi(G) \ge 4$.
\end{lemma}

\section{A Simple Resolution Lower Bound for 4-Colourability on Sparse
  Graphs}
\label{sec:resolution-lb}

In this section we reprove the main result of Beame et
al.~\cite{BCCM05RandomGraph}, namely that resolution requires large
width to refute the claim that random graphs have small chromatic
number. We hope that the exposition below may serve as a gentle
introduction to the polynomial calculus lower bounds that will follow
later, which build on similar albeit slightly more complicated
concepts. Readers only interested in the polynomial calculus lower
bounds may safely skip this section.

For expert readers let us remark that while Beame et
al.~\cite{BCCM05RandomGraph} build on the width lower bound
methodology introduced by Ben-Sasson and
Wigderson~\cite{BW01ShortProofs}, we rely on the game characterization
due to Pudlák~\cite{Pudlak00ProofsAsGames}. We recover the bounds of
Beame et al.~precisely for $k \geq 4$, but for $k=3$ we incur a slight
loss in parameters.

In the following we recollect some standard notions to then prove the
resolution width lower bound in \refsec{sec:res-lb}.

\subsection{Graph Colouring in CNF and the Resolution Proof System}
\label{sec:resolution}

A \emph{clause} is a disjunction over a set of \emph{literals}
$\bigvee_{i \in S}\ell_i$, where every literal $\ell_i$ is either a
Boolean variable $x$ or the negation $\olnot x$ thereof.  The
\emph{width} of a clause is the number of literals in it. A
\introduceterm{CNF formula}
$F = \clc_1 \land \formuladots \land \clc_m$ is a conjunction of
clauses, and the \emph{width} of $F$ is the maximum width of any
clause $\clc_i$ in $F$.

A \introduceterm{resolution refutation} of a CNF formula~$F$ is as an
ordered sequence of clauses
$\proofstd = (\cld_1, \dotsc, \cld_{\tau})$ such that
$\cld_{\tau} = \emptycl$ is the empty clause and each~$\cld_i$ either
occurs in $\formf$ or is derived from clauses $\cld_{j_1}$ and
$\cld_{j_2}$, with $j_1 < i$ and $j_2 < i$, by the
\introduceterm{resolution rule}
\begin{equation}
\AxiomC{$\clb \lor x$}
\AxiomC{$\clc \lor \olnot{x}$}
\BinaryInfC{$\clb \lor \clc$}
\DisplayProof
\eqperiod
\end{equation}
The \introduceterm{width} of a resolution refutation
$\proofstd = (\cld_1, \dotsc, \cld_{\tau})$ is the maximum width of any
clause $\cld_i$ in $\proofstd$ and the \introduceterm{length} of
$\proofstd$ is~$\tau$. The \emph{size-width}
relation~\cite{BW01ShortProofs} relates the minimum width~$W$ required
to refute a CNF formula~$F$ to the minimum length of any resolution
refutation of $F$: if $F$ is of initial width~$w$ and defined over $n$
variables, then any resolution refutation of $F$ requires size
$\exp\bigl(\Omega\big((W - w)^2/n\big)\bigr)$.

As resolution operates over clauses, in contrast to polynomial
calculus that operates over polynomials, we need to define the
colourability formula as a CNF formula. In this section, for a graph
$G$ and integer $k$, we let $\Col{G}{k}$ denote the CNF formula
consisting of the clauses
\begin{subequations}\label{eq:cnf-encoding}
  \begin{align}\label{eq:cnf-encoding1}\bigvee_{i=1}^\kcolourcons x_{v, i}
    &, \qquad v \in V(G)
    &&[\text{every vertex is assigned a colour}]\\
      \label{eq:cnf-encoding2}
    \olnot{x}_{v, i} \lor \olnot{x}_{v, i'}
    &, \qquad v \in V(G),\ i \neq i'
    &&[\text{no vertex gets more than one colour}] \\
      \label{eq:cnf-encoding3}
    \olnot{x}_{u, i} \lor \olnot{x}_{v, i}
    &, \qquad (u, v) \in E(G), \ i \in [k].
    &&[\text{no two adjacent vertices get the same colour}]
  \end{align}   
\end{subequations}
Clearly, the CNF formula $\Col{G}{k}$ is satisfiable if and only if
$G$ is $k$-colourable.

\subsection{Resolution Lower Bounds}
\label{sec:res-lb}

For the sake of simplicity we prove the theorem below for the
$4$-colourability formula only. The theorem can be extended to
$3$-colourability using concepts from \refsec{sec:closure}.

\begin{theorem}
  \label{thm:res-lb-sparse}
  Let $G$ be a graph, and suppose that $\ell \in \N^+$ and $\eps > 0$
  are such that $G$ is $(\ell, 1/2-\eps)$-sparse. Then every
  resolution refutation of $\Col{G}{4}$ requires width at least
  $\ell/4$.
\end{theorem}

Combining \reflem{lem:razborov-sparse} with the above theorem we
obtain the following corollaries for random graphs.

\begin{corollary}[\cite{BCCM05RandomGraph}]
  \label{cor:res-lb-gndn}
  For any $\eps > 0$, if $G$ is a graph sampled from $\gndn$, then
  \aas every resolution refutation of $\Col{G}{4}$ requires width
  $n/4(4d)^{3+\epsilon}$.
\end{corollary}

While Beame et al.~\cite{BCCM05RandomGraph} do not consider random
regular graphs, it is not hard to see that their techniques can be
used to prove the following statement.

\begin{corollary}
  \label{cor:res-lb-gnd}
  For any $\eps > 0$, if $G$ is a graph sampled from $\gnd$ and
  $d^2 = \littleoh{\log n}$, then \aas every resolution refutation of
  $\Col{G}{4}$ requires width $n/4(8d)^{3+\epsilon}$.
\end{corollary}

We prove \refthm{thm:res-lb-sparse} using the prover-adversary game as
introduced by Pudl\'{a}k~\cite{Pudlak00ProofsAsGames}, which we
describe here adapted to the colouring formula. The width-$w$
prover-adversary game for colouring proceeds in rounds. In each round
the prover either queries or forgets the colouring of a vertex. In
response to the former, the adversary has to respond with a colouring
of the queried vertex. The prover has limited memory and may only
remember the partial colouring of up to $w$ vertices.  The prover wins
whenever the remembered partial colouring is improper. This game
characterises resolution refutation width while in our setting a more
precise statement is that the prover has a winning strategy in the
width-$w$ prover-adversary game for colouring if and only if there is
a resolution refutation of the colourability formula where every
clause in the refutation mentions at most $w$ vertices.

Ultimately we want to design a strategy for the adversary so that the
prover cannot win with limited memory. In order not to reach a partial 
colouring that is impossible to extend to the newly queried vertex, 
whenever the prover remembers a
partial colouring defined on a subset $U$ of the vertices, the
adversary maintains a consistent partial colouring defined
on a \emph{closure} of~$U$, as defined below.
Intuitively, a closure of~$U$ 
should be a slightly larger set that contains~$U$ and other vertices 
that, if not taken into account when colouring~$U$, 
might not be possible to properly colour later on. 

\begin{definition}[Closure]
  Let $G=(V,E)$ be a graph and let $U \subseteq V$. We say that $U$ is
  \emph{closed} if all vertices $v \in V\setminus U$ satisfy
  $\setsize{N_U(v)} \leq 1$. A \emph{closure} of $U$ is any minimal
  closed set that contains $U$.
\end{definition}

\begin{proposition}
  \label{clm:res-closure-single}
  Every set of vertices has a unique closure.
\end{proposition}
\begin{proof}
  Suppose there are two distinct closures $W_1$ and $W_2$ of a set $U$. As
  both $W_1$ and $W_2$ contain $U$ it holds that
  $U \subseteq W_\cap = W_1 \cap W_2$. By minimality of a closure, the
  set $W_\cap$ is not a closure of $U$ and hence there is a vertex
  $v$ not in $W_\cap$ such that
  $\setsize{N_{W_\cap}(v)} \geq 2$. But this implies that the set
  $W_i$ satisfying $v \not \in W_i$ is not closed, which contradicts
  the initial assumption.
\end{proof}

In light of Proposition~\ref{clm:res-closure-single}, for a set $U$ we
write $\closure{U}$ to denote the unique closure of $U$. We now show
that if the underlying graph $G$ is sparse, then the closure of a set
$U$ is not much larger than $U$ itself.

\begin{lemma}
  \label{lem:res-closure-small}
  Let $G$ be an $(\ell, 1/2-\eps)$-sparse graph for $\ell \in \N^+$
  and $\eps > 0$. Then any set $U \subseteq V$ of size at most
  $\ell/4$ satisfies
  $\setsize{\closure{U}} \leq 4\setsize{U} \leq \ell$.
\end{lemma}

\begin{proof}
  Let us consider the following process: set $U_0 = U$ and
  $i=0$. While there is a vertex $v \in V\setminus U_i$ satisfying
  $\setsize{N_{U_i}(v)} \geq 2$, set $U_{i+1} = U_i \cup \set{v}$ and
  increment $i$.

  Suppose this process terminates after $t$ iterations. Clearly the
  final set $U_t$ contains $U$ and is closed. Furthermore it is a
  minimal set with these properties and is hence the closure of $U$.
  As we add at least $2$ edges to $E(U_i)$ in every iteration it holds
  that $\setsize{E(U_i)} \geq 2i$ and, as we add a single vertex in
  each iteration, we have $\setsize{U_i} = \setsize{U} + i$. Suppose $t \geq 3\setsize{U}$. In iteration $i = 3\setsize{U}$ it
  holds that $\setsize{E(U_i)} \geq 6\setsize{U}$, while
  $\setsize{U_{i}} = 4\setsize{U} \leq \ell$. Hence
  $\setsize{E(U_i)} \geq 3\setsize{U_i}/2$, which contradicts the
  assumption on sparsity. We may thus conclude that the closure of a
  set $U$ is of size at most $4\setsize{U}$, as claimed.
\end{proof}

In order to design a strategy for the adversary, we prove that it
is always possible to extend a colouring on a small closed
set of vertices to a slightly larger set.

\begin{lemma}
  \label{lem:res-extend-coloring}
  Let $G=(V,E)$ be an $(\ell, 1/2-\eps)$-\nobreakdash sparse graph,
  let $U \subseteq V$ be a closed set of size at most $\ell$ and
  suppose $\chi$ is a proper $4$-colouring of $G[U]$. Then $\chi$ can
  be extended to $G[W]$ for any set $W \supseteq U$ of size at most
  $\setsize{W} \leq \ell$.
\end{lemma}
\begin{proof}
  By induction on $s = \setsize{W\setminus U}$. The statement clearly
  holds for $s=0$. For the inductive step, let $v \in W\setminus U$ be
  such that $\setsize{N_{W\setminus U}(v)} \leq 2$, as guaranteed to
  exist by the assumption on sparsity since it holds that
  $\setsize{E(W\setminus U)}\leq (3/2 - \eps)|W\setminus U|$. By the
  inductive hypothesis there is an extension of $\chi$ to the set
  $W \setminus \set{v}$.
  
  As $U$ is closed, the vertex $v$ has at most one neighbour in
  $U$. Hence it holds that $\setsize{N_{W}(v)}\leq 3$ and we may thus
  extend the colouring to $v$.
\end{proof}

With Lemmas~\ref{lem:res-closure-small}
and~\ref{lem:res-extend-coloring} at hand we are ready to prove
\refthm{thm:res-lb-sparse}.

\begin{proof}[Proof of \refthm{thm:res-lb-sparse}]
  Let us describe a strategy for the adversary. At any point in the
  game the adversary maintains a partial $4$-colouring $\chi$ to the
  closure of the vertices $U$ the prover remembers. Whenever the
  prover queries a vertex in the closure the adversary responds
  accordingly. If a vertex $v$ outside the closure of $U$ is queried,
  then the adversary extends $\chi$ to $\closure{U \cup \set{v}}$ by
  virtue of \reflem{lem:res-extend-coloring} and responds
  accordingly. Finally, if the prover forgets the colouring of a
  vertex $u \in U$, then we shrink our closure to the closure of
  $U \setminus \set{u}$. Here we use the fact that by minimality of
  the closure it holds that
  $\closure{U} \supseteq \closure{U \setminus \set{u}}$.

  As, by \reflem{lem:res-closure-small}, the closure of a set $U$ is
  at most a factor $4$ larger than $U$ and, by
  \reflem{lem:res-extend-coloring}, we may maintain a valid
  $4$-colouring to the closure as long as the size of the closure is
  bounded by $\ell$, the prover cannot win the game if
  $w \leq \ell/4$. Therefore, every resolution refutation of the
  $4$-colourability formula defined over an
  $(\ell, 1/2 - \eps)$-sparse graph contains a clause that mentions at
  least $\ell/4$ vertices and hence has width at least $\ell/4$.
\end{proof}

Combining Theorems~\ref{cor:res-lb-gndn} and~\ref{cor:res-lb-gnd} with
the mentioned size-width relation we obtain optimal $\exp(\Omega(n))$
resolution size lower bounds for constant $d$.

\section{Polynomial Calculus
  Lower Bounds: The General Framework}
\label{sec:techniques}

In this section we provide a proof overview of \refth{th:main-theorem}
and then revisit the general framework, as introduced by Alekhnovich
and Razborov~\cite{AR03LowerBounds}, for obtaining polynomial calculus
lower bounds.

\subsection{Proof Overview}
\label{sec:proof-overview}

As outlined in the introduction, the construction of our
pseudo-reduction $\razborovoperator$ for the colouring formula follows
the general paradigm introduced by~\cite{AR03LowerBounds} which has
been used in several subsequent
papers~\cite{GL10Automatizability,GL10Optimality,MN15GeneralizedMethodDegree}.
The idea is that given an initial set of polynomials $\pcsp$, for
every monomial~$\pcmonm$ we identify a subset $\Smap{\pcmonm}$ of
polynomials in~$\pcsp$ that are in some sense relevant to~$\pcmonm$.
Then we define $\razborovoperator$ on the monomial $m$ as the
reduction modulo the ideal~$\idealS{\pcmonm}$ generated by the
polynomials, and extend $\razborovoperator$ linearly to arbitrary
polynomials. The goal is to show that $\razborovoperator$ satisfies
properties \ref{item:1}-\ref{item:3} in \reflem{lem:r-operator}, which
typically requires showing that $S$ satisfies two main technical
properties.  The first property is captured by what we call a
\emph{\satlemma}, which states that if the degree of $m$ is at most
some parameter $D$, then the associated set~$\Smap{\pcmonm}$ is
satisfiable and is in some sense well-structured.
The second property is described by the \emph{reducibility lemma},
which states that for every ideal $I$ that is generated by a
well-structured set of polynomials that contains $\Smap{\pcmonm}$ and
is satisfiable, it holds
that~$R_{I}(m) = R_{\langle \Smap{\pcmonm} \rangle}(m)$.  Once these
lemmas are in place, and as long as $S$ satisfies some other simple
conditions, a degree lower bound of $D$ follows by an argument
presented in~\cite{AR03LowerBounds}.

We formalise these arguments in \reflem{lem:ar-method} by extracting
the essential parts of the Alekhnovich--Razborov~\cite{AR03LowerBounds}
framework and making explicit the properties the map $S$ must satisfy
in order for the proof of the lower bound to go through.  Apart from
the properties in \refdef{def:acceptable} capturing some type of
multilinearity, monotonicity and closedness of $S$ and ensuring that
axioms with leading monomial $m$ are in $S(m)$, we introduce in
\reflem{lem:ar-method} two sufficient conditions for the lower bound
to follow: the first corresponds to the \satlemma and the second to
the reducibility lemma. With this in hand, the ensuing sections can
focus on defining the map $S$ and proving that it satisfies the
conditions, without having to refer to reduction operators.

\subsection{Revisiting the Alekhnovich--Razborov Framework} 
\label{sec:ar-method}

We now revisit the general framework introduced by Alekhnovich and
Razborov~\cite{AR03LowerBounds} for proving polynomial calculus lower
bounds.  We extract from their proofs a formal statement that if the
function $S$, mapping monomials to relevant subsets of axioms,
satisfies two conditions, namely the \emph{\satcondition} and the
\emph{reducibility conditions}, along with some natural conditions
as discussed in the following, then this implies polynomial calculus
degree lower bounds.

Before discussing the additional conditions that the mapping~$S$ needs
to satisfy, we need to fix an \emph{admissible} order $\prec$ on the
monomials in~$\F[x_1,\ldots,x_n]$ as defined in
\refsec{sec:algebra}. Recall that an admissible order is a total order
that respects multiplication, that is, if~${m_1 \prec m_2}$ then it
holds that~$mm_1 \prec mm_2$. The \emph{leading monomial} of a
polynomial~$p$ is the largest monomial according to~$\prec$ that
appears in~$p$. For the remainder of this section we implicitly assume
that monomials in~$\F[x_1,\ldots,x_n]$ are ordered according to an
admissible order $\prec$. Let us further adopt the convention that
whenever we write a polynomial $p$ as a sum of monomials
$p = \sum_i a_i m_i$, then $a_i \in \F$ are field elements and the
$m_i \in \F[x_1,\ldots,x_n]$ are monomials ordered such that
$m_{i} \prec m_j$ for all $j < i$. In particular, $m_1$ is the leading
monomial of $p$.

Our first additional requirement on $S$ is that it maps monomials
according to the variables in the monomial, so that if two monomials
$m$ and $m'$ both contain the same set of variables, then
$S(m) = S(m')$.  We further require that $S$ is in some sense
monotone, namely, for any variable~$x$, we require that if
$S(m')\subseteq \Smap{\pcmonm}$ for monomials $m'\prec m$, then it
also holds that $S(xm')\subseteq S(xm)$. Moreover, the image of $S$
should consist of sets that are ``closed'' in the sense that for any
$m$, the set of axioms $S(m)$ contains the union of all $S(m')$ where
$m'\prec m$ and $\vars{m'}\subseteq\vars{S(m)}$.  Finally, we require
that if~$m$ is the leading monomial of an axiom~$p$, then
$p \in S(m)$.  These properties are formalised in the following
definition. 

\begin{definition}[\expandafter\MakeUppercase\acceptablename]
  \label{def:acceptable}
  Let $\pcsp$ be a set of polynomials over
  $\F[x_1, \ldots, x_\nvars]$, let $\prec$ be an admissible order on
  the monomials in~$\F[x_1,\ldots,x_n]$, and let
  $S\colon 2^{\set{x_1, \ldots, x_\nvars}} \rightarrow 2^{\pcsp}$ be a
  function that maps subsets~$Y \subseteq \set{x_1, \ldots, x_\nvars}$
  of variables to subsets~$S(Y) \subseteq \pcsp$ of polynomials.  For
  a monomial~$m$ we write~$S(m)$ for $S\bigl(\vars{m}\bigr)$ and we
  say that $S$ is a \emph{\acceptable{\pcsp}} if the following holds.
  \begin{enumerate}
  \item For every variable $x$ and for all monomials $m$ and $m'$ such
    that $m'\prec m$, it holds that if
    $S(m')\subseteq \Smap{\pcmonm}$, then $S(xm')\subseteq S(xm)$.
    \label{it:property-Sxm}
    
  \item For all monomials $m$ and $m'$ such that $m'\prec m$, it holds
    that if $\vars{m'}\subseteq\Vars{\Smap{\pcmonm}}$, then
    $S(m')\subseteq \Smap{\pcmonm}$.
    \label{it:property-Sm}
    
  \item For all $p\in \pcsp$, it holds that $p\in S(m)$, where $m$ is
    the leading monomial in $p$.
    \label{it:property-axiom}
  \end{enumerate}  
\end{definition}

If the set of polynomials $\pcsp$ is not essential we call a
\acceptable{\pcsp} simply a \emph{\acceptablename}. Let us record four
technical claims that follow from the above properties of a
\acceptablename\xspace and that we will later use to prove that an
appropriate map~$S$ implies a polynomial calculus degree lower
bound. 

The first claim is that if $m_1$ is the leading monomial of an axiom
$p$, not only does $S(m_1)$ contain~$p$ but it also contains $S(m_i)$
for all monomials $m_i$ that appear in $p$.

\begin{claim}\label{cl:property-p}
  If $S$ is a \acceptable{\pcsp}, then for all $p\in \pcsp$, if
  ${\pcpolyp = \sum_i a_i \pcmonm_i}$, then
  $\set{p} \union \Union_i \bigl(S(m_i)\bigr) \subseteq S(m_1)$.
\end{claim}
\begin{proof}
  By property~\ref{it:property-axiom} of \refdef{def:acceptable} it
  holds that $p\in S(m_1)$. This implies that
  $\vars{m_i} \subseteq \vars{p}$ is contained in $\Vars{S(m_1)}$ and
  thus for $i\neq 1$ since $m_i\prec m_1$,
  property~\ref{it:property-Sm} of \refdef{def:acceptable} implies
  that $S(m_i) \subseteq S(m_1)$.
\end{proof}

The next claim states that if the variables of $m'$ is a subset of the
variables of $m$, then $S(m')\subseteq S(m)$.

\begin{claim}\label{cl:property-multilinear}
  If $S$ is a \acceptablename\xspace, then for all monomials
  $m$ and $m'$ such that $\vars{m'}\subseteq\vars{m}$ it holds that
  $S(m') \subseteq \Smap{\pcmonm}$.
\end{claim}
\begin{proof}
  Note that $S(1) \subseteq S(m)$ by property~\ref{it:property-Sm} of
  \refdef{def:acceptable}, where we use that $1 \prec m$.  Thus, by
  inductively applying property~\ref{it:property-Sxm} of
  \refdef{def:acceptable} and using the fact that $\prec$ respects
  multiplication, we obtain that $S(m') \subseteq S(m'\cdot
  m)$. Finally, since $\vars{m'}\subseteq\vars{m}$, it holds
  that~$\vars{m'} = \vars{m' \cdot m}$ and
  hence~$S(m) = S(m'\cdot m) \supseteq S(m') $.
\end{proof}

We also prove that monomials $m'$ that appear when $m$ is reduced
modulo $S(m)$ must be such that $S(m') \subseteq S(m)$ and that
$S(xm') \subseteq S(xm)$.

\begin{claim}
  \label{cl:reduction-S}
  If $S$ is a \acceptablename\xspace, then for all monomials $m$ and
  $m'$ such that $m'$ appears in
  $R_{\langle \Smap{\pcmonm} \rangle}(\pcmonm)$ it holds
  that~$\Smap{\pcmonm'} \subseteq \Smap{\pcmonm}$ and
  $\Smap{x\pcmonm'} \subseteq \Smap{x\pcmonm}$ for any variable $x$.
\end{claim}
\begin{proof}
  We assume $m'\neq m$, and thus $m'\prec m$, otherwise the claim
  follows trivially.  Moreover, we observe that it suffices to show
  that $S(m')\subseteq S(m)$, since we can then use
  property~\ref{it:property-Sxm} of \refdef{def:acceptable} to
  conclude that $S(xm') \subseteq S(xm)$.

  We now argue
  that~$\vars{m'} \subseteq \vars{m} \union \Vars{\Smap{\pcmonm}}$.
  Towards contradiction, suppose this is not the case. Let $\rho$
  denote the assignment that assigns all variables in~$m'$ that are
  not in~$\Vars{\Smap{\pcmonm}} \cup \vars{{\pcmonm}}$ to 0.
  Recall that $R_{\langle \Smap{\pcmonm}\rangle}(\pcmonm)$ denotes the
  reduction of $m$ modulo $\langle \Smap{\pcmonm}\rangle$ and that
  there is a unique representation
  $\pcmonm = q + R_{\langle \Smap{\pcmonm}\rangle}(\pcmonm)$ with
  $q\in \langle S(m)\rangle$.
  No variable in~$\pcmonm$ nor in any of the generators
  of~$\idealS{\pcmonm}$ is assigned by~$\rho$,
  so~$\restrict{\pcmonm}{\rho} = \pcmonm$
  and~${\restrict{q}{\rho} \in \idealS{\pcmonm}}$. Moreover, we
  have~$\restrict{\pcmonm'}{\rho} = 0$
  so~${\restrict{R_{\idealS{\pcmonm}}(m)}{\rho} \neq R_{\langle
      \Smap{\pcmonm}\rangle}(m)}$. But this contradicts that the
  representation
  ~$\pcmonm = q + R_{\langle \Smap{\pcmonm}\rangle}(\pcmonm)$ is
  unique, and
  thus~$\vars{m'} \subseteq \vars{{\pcmonm}} \cup
  \Vars{\Smap{\pcmonm}}$.

  Write~$m' = m_1'\cdot m_2'$ such
  that~${\vars{m_1'} \subseteq \vars{m}}$
  and~${\vars{m_2'} \subseteq \Vars{\Smap{\pcmonm}}}$. Observe that
  $m'_2 \prec m$ since $1 \prec m' \prec m$ and the order $\prec$
  respects multiplication as it is admissible. Hence by
  property~\ref{it:property-Sm} of \refdef{def:acceptable} applied to
  $m_2'$ and $m$ it follows that~$S(m_2') \subseteq
  \Smap{\pcmonm}$. By iteratively applying property~\ref{it:property-Sxm} of
  \refdef{def:acceptable} for each $x \in \vars{m_1'}$, and again
  using that the order $\prec$ respects multiplication, we can
  conclude
  that~${\Smap{\pcmonm'} = S(m_1'\cdot m_2') \subseteq S(m_1'\cdot
    m)}$. Finally, since~${\vars{m_1'} \subseteq \vars{m}}$, we have
  that~${S(m_1' \cdot m) = S(m)}$, and
  therefore~$\Smap{\pcmonm'} \subseteq \Smap{\pcmonm}$.
\end{proof}

From this claim we can deduce that if irreducibility is preserved
modulo some larger ideal, then the reduced polynomials are the same as
done in the following.

\begin{claim}
  \label{cl:reduction-same}
  Let $S$ be a \acceptable{\pcsp} and let
  $\subsetcalP\subseteq\pcsp$. Suppose that every monomial $m$
  irreducible modulo $\langle S(m)\rangle$ that satisfies
  $S(m) \subseteq \subsetcalP$ is also irreducible modulo
  $\langle \subsetcalP \rangle$. Then for all monomials $m'$ such that
  $S(m') \subseteq \subsetcalP$ it holds that
  $R_{\langle \subsetcalP \rangle}(\pcmonm') = R_{\langle S(\pcmonm')
\rangle}(\pcmonm') $.
\end{claim}
\begin{proof}
  Let $m'$ be an arbitrary monomial satisfying
  $S(m') \subseteq \subsetcalP$ and suppose that any irreducible
  monomial $m$ modulo $\langle S(m)\rangle$ such that
  $S(m) \subseteq \subsetcalP$ is also irreducible modulo
  $\langle \subsetcalP \rangle$. We want to prove that
  $R_{\langle \subsetcalP \rangle}(\pcmonm') = R_{\langle S(\pcmonm')
    \rangle}(\pcmonm') $.

  Let us write
  $R_{\langle \Smap{\pcmonm'} \rangle}(\pcmonm') = \sum_{i} a_i m_i$.
  Note that by definition each such $m_i$ is irreducible
  modulo~$\langle \Smap{\pcmonm'} \rangle$. As by
  \refclaim{cl:reduction-S} it holds that $S(m_i) \subseteq S(m')$ we
  may conclude that each $m_i$ is irreducible modulo
  $\langle S(m_i)\rangle$. Thus, by assumption along with the
  inclusions $S(m_i) \subseteq S(m') \subseteq \subsetcalP$, each
  $m_i$ is irreducible modulo~$\langle \subsetcalP\rangle$. This
  implies that
  $R_{\langle \Smap{\pcmonm'} \rangle}(\pcmonm') = R_{\langle
    \subsetcalP \rangle}\bigl(R_{\langle \Smap{\pcmonm'}
    \rangle}(\pcmonm')\bigr) = R_{\langle \subsetcalP\rangle}(m')$,
  using once more that $S(m') \subseteq \subsetcalP$.
\end{proof}

We can now make the formal claim that polynomial calculus degree lower
bounds follow from the existence of a \acceptablename\xspace$S$
satisfying the two conditions corresponding to the \satlemma and the
reducibility lemma. This fact is implicit in~\cite{AR03LowerBounds}
and we make it explicit below. It is worth remarking that
Filmus~\cite{Filmus19AlekhnovichRazborov} establishes a lemma of
similar flavour. However, in contrast to Filmus, we do not introduce
another layer of abstraction but rather state the lemma directly in
terms of $S$.

\begin{lemma}
  \label{lem:ar-method}
  Let $D \in \N^+$, let $\pcsp$ be a set of 
  polynomials over $\F[x_1, \ldots, x_\nvars]$ of degree at most $D$,
  and denote by $S$ a \acceptable{\pcsp}. If all monomials
  $m$ of degree at most $D$ satisfy that
  \begin{enumerate}
  \item \emph{Satisfiability condition:} the set of axioms $S(m)$ is
    satisfiable; and
  \item \emph{Reducibility condition:} for all monomials $m'$, if $\Smap{m'} \subseteq S(m)$,
    then $m'$ is reducible modulo $\langle S(m') \rangle$ if and only
    if $m'$ is reducible modulo $\langle \Smap{m} \rangle$;
  \end{enumerate}
  then any polynomial calculus refutation of $\pcsp$ over $\F$
  requires degree strictly greater than $D$.
\end{lemma}

\begin{proof}
  Let $\razborovoperator$ be the operator defined on monomials~$m$ by
  $\razborovoperator(m) = R_{\langle \Smap{\pcmonm} \rangle}(m)$ and
  extended by linearity to polynomials. Our goal is to show
  that~$\razborovoperator$ is a degree-$D$ pseudo-reduction for
  $\pcsp$. That is, according to \refdef{def:r-op}, we need to show
  that~$\razborovoperator(1) = 1$, that~$\razborovoperator$ maps all
  polynomials in $\pcsp$ to $0$, and that for every monomial~$m$ of
  degree at most $D-1$ and every variable $x$, it holds
  that~$\razborovoperator(x m) = \razborovoperator(x
  \razborovoperator(m))$. If we manage to show these properties, then
  we can appeal to \reflem{lem:r-operator} to reach the desired
  conclusion.

  By the \satcondition condition we have that $S(1)$ is satisfiable, 
  that is, the set of polynomials $S(1)$ has
  a common root. This implies that the entire ideal
  $\langle S(1) \rangle$ has a common root and, therefore, $1$ is not 
  in $\langle S(1) \rangle$ and hence is not
  reducible to $0$ modulo $\langle S(1) \rangle$. As $1$ is the
  smallest monomial in the order we conclude that
  $\razborovoperator(1) = R_{\langle S(1) \rangle}(1) = 1$.
  
  To see that~$\razborovoperator$ maps each polynomial
  $\pcpolyp \in \pcsp$ to $0$, let
  ${\pcpolyp = \sum_j a_j \pcmonm_j}$.  By \refclaim{cl:property-p} it
  follows that $p\in S(m_1)$ and that $S(m_j)\subseteq S(m_1)$ for all
  $j$.  It therefore holds that
  \begin{subequations}
    \begin{align}
      \razborovoperator(\pcpolyp)
      &= \sum_j a_j R_{\idealS{\pcmonm_j}}(\pcmonm_j)
      &&[\text{by definition of $\razborovoperator$ and $p$}]\\
      &= \sum_j a_j R_{\langle S(m_1) \rangle}(\pcmonm_j)
      &&[\text{by \refclaim{cl:reduction-same} and the reducibility condition}]\\
      &= R_{\langle S(m_1) \rangle}\Big(\sum_j a_j \pcmonm_j\Big)
      &&[\text{by linearity of $R_{\langle S(m_1) \rangle}$}]\\
      &= R_{\langle S(m_1) \rangle}(\pcpolyp)
      &&[\text{by definition of $p$}] \\
      &= 0\eqcomma &&[\text{since $\pcpolyp\in S(m_1)$}]
    \end{align} 
  \end{subequations}
  where we note that in order to apply the reducibility condition we
  use the assumption that $p$ has degree at most $D$.
  
  Finally, we need to show that for every monomial $\pcmonm$ of degree
  at most $\mdegreestd-1$ and every variable~$\indeterminate$, it
  holds
  that~${\razborovoperator(\indeterminate\pcmonm) =
    \razborovoperator(\indeterminate \razborovoperator(\pcmonm))}$.
  Let
  ${\razborovoperator(\pcmonm) = R_{\langle S(m)\rangle}(m) = \sum_j
    a_j \pcmonm_j}$. By definition of $\razborovoperator$ we have that
  \begin{equation}
    \razborovoperator(\indeterminate \razborovoperator(\pcmonm))
    =
    \sum_{j} a_j \razborovoperator(\indeterminate \pcmonm_j)
    =
    \sum_{j} a_j
    R_{\idealS{\indeterminate \pcmonm_j}}(\indeterminate\pcmonm_j)
    \eqperiod
    \label{eq:not-hard-step} 
  \end{equation}
  We now argue that since $\pcmonm$ is a monomial of degree at most
  $D-1$, reducing modulo $\idealS{\indeterminate\pcmonm_j}$ or
  $\idealS{\indeterminate\pcmonm}$ results in the same polynomial.
  More formally, we claim that
  \begin{equation}
    \label{eq:hard-step}
    R_{\idealS{\indeterminate \pcmonm_j}}(\indeterminate\pcmonm_j)
    =  R_{\idealS{\indeterminate \pcmonm}}(\indeterminate\pcmonm_j) \end{equation}
  follows from \refclaim{cl:reduction-same} with
  $\subsetcalP = \Smap{\indeterminate \pcmonm}$ and $m' =
  x\pcmonm_j$. Let us check that the conditions of
  \refclaim{cl:reduction-same} are satisfied. We need to check that
  \begin{enumerate}
  \item any monomial $\tilde m$ irreducible modulo
    $\langle S(\tilde m) \rangle$ satisfying that
    $S(\tilde m) \subseteq \Smap{\indeterminate \pcmonm}$ is also
    irreducible modulo $\langle \Smap{\indeterminate \pcmonm}
    \rangle$, and \label{it:S-red}
  \item
    $\Smap{\indeterminate\pcmonm_j} \subseteq
    \Smap{\indeterminate\pcmonm}$.\label{it:subset}
  \end{enumerate}
  The latter follows immediately from \refclaim{cl:reduction-S} and
  \refitem{it:S-red} follows from the reducibility condition: since
  the monomial $xm$ has degree at most~$D$, it holds for all monomials
  $\tilde m$, if $S(\tilde m) \subseteq S(xm)$, then $\tilde m$ is
  irreducible modulo $\langle S(\tilde m)\rangle$ if and only if
  $\tilde m$ is irreducible modulo $\langle S(xm)\rangle$. We may thus
  conclude \refeq{eq:hard-step} from \refclaim{cl:reduction-same}
  applied with $\subsetcalP = \Smap{\indeterminate \pcmonm}$ and
  $m' = x\pcmonm_j$.
  
  We finish the
  proof of \reflem{lem:ar-method} by noting that
  \begin{subequations}
    \begin{align}   
      \razborovoperator\bigl(\indeterminate \razborovoperator(\pcmonm)\bigr) 
      &= 
        \sum_{j} a_j R_{\idealS{\indeterminate \pcmonm}}(\indeterminate\pcmonm_j)
      &&[\text{by \refeq{eq:not-hard-step} and \refeq{eq:hard-step}}]\\
      &= R_{\idealS{\indeterminate\pcmonm}} \Big(\sum_{j} a_j \indeterminate \pcmonm_j\Big)
      &&[\text{by linearity of $R_{\langle S(xm) \rangle}$}]\\
      &= R_{\idealS{\indeterminate\pcmonm}} \bigl(\indeterminate R_{\idealS{\pcmonm}}(\pcmonm)\bigr)
      &&[\text{as $R_{\langle S(m)\rangle}(\pcmonm) = \textstyle\sum_j a_j \pcmonm_j$} ]\\
      &= R_{\idealS{\indeterminate\pcmonm}} (\indeterminate \pcmonm)
      &&[\text{by \refobs{obs:ideal-contain-reduction}, using \refclaim{cl:property-multilinear}}] \\
      &= \razborovoperator(\indeterminate \pcmonm) \eqperiod
    \end{align} 
  \end{subequations}
  This establishes that~$\razborovoperator$ is a degree-$D$
  pseudo-reduction for $\pcsp$ as defined in \refdef{def:r-op}. We may
  thus appeal to \reflem{lem:r-operator} to obtain the desired
  polynomial calculus refutation degree lower bound for $\pcsp$. This
  completes the proof.
\end{proof}
 
\section{Closure for Graph Colouring}
\label{sec:closure}
\label{sec:rt22techniques}

While the previous section was fairly generic and we made few
assumptions on the polynomial system $\pcsp$ to be refuted, we now
turn our attention to the colouring formula and make the first steps
towards defining a support $S$ for it. The goal of this section is to
introduce the \emph{closure} of a vertex set, a crucial notion in the
definition of a support $S$ for the colouring formula.
  
In order to state the definition of our closure we need to introduce
some terminology. Let $G=(V, E)$ be a graph. A \emph{walk of length
  $\tau$} in $G$ is a tuple of vertices~$(v_1, \ldots, v_{\tau+1})$
satisfying $(v_i, v_{i+1} ) \in E$ for all $i\in [\tau]$.  A
\emph{simple path} is a walk where all vertices are distinct and a
\emph{simple cycle} is a walk $(v_1, \ldots, v_{\tau+1})$ of length
$\tau\geq 3$ where $v_1 = v_{\tau+1}$ and all other vertices are
distinct.

Suppose the set of vertices $V$ has a linear order $\prec$ on $V$. An
\emph{increasing (decreasing) path} in~$G$ is a simple path
$(v_1, \ldots, v_\tau)$ where $v_{i} \prec v_{i+1}$
($v_{i} \succ v_{i+1}$) for all $i\in [\tau-1]$. For vertices
$u, v\in V$ we say that~$v$ is a \emph{descendant of $u$} if there
exists a decreasing path from $u$ to $v$, and for a set of vertices
$U \subseteq V$ we say that $v$ is a \emph{descendant of $U$} if it is
a descendant of some vertex in $U$. We write $\Desc{U}$ to denote the
set of all descendants of~$U$. We define every vertex to be a
descendant of itself so that $U \subseteq \Desc{U}$.

The definition of our closure, besides the notion of descendants, also
requires the notions of a \emph{hop} and a \emph{lasso}: a
\emph{$\hoplength$\nobreakdash-hop \wrt a set $U \subseteq V$} is a
simple path or a simple cycle of length $\hoplength$ with the property
that the two endpoints are both contained in~$U$ (in the case of
cycles, the two endpoints coincide), while all other vertices are not
in~$U$. Similarly a \emph{lasso \wrt $U$} is a paw graph, that is, a
walk $(v_1, v_2, v_3, v_4, v_5)$ with $v_2 = v_5$ and all other
vertices being distinct, such that only $v_1$ is in $U$. See
\reffig{fig:hop} for an example of a set with some hops and a lasso.

\begin{figure}
  \centering
  \includegraphics{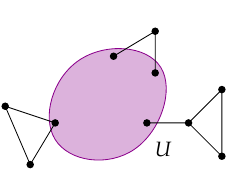}
  \caption{A set $U$ with a 3-hop to the left, a 2-hop on top and a
    lasso to the right.}
  \label{fig:hop}
\end{figure}

With these notions in place we can now define 
a closure of a set $U\subseteq V$.
After showing the closure is unique, 
we define a process that given $U$ constructs
its closure and then
prove that in sparse graphs the closure of~$U$ 
satisfies two important properties:
\begin{enumerate}
\item it is the set of descendants of a set that is not much larger
  than~$U$, as long as $U$ itself is not too large, and,
\item if we remove the closure from the graph, then any small enough
  vertex set has a specially structured proper $3$-colouring.
\end{enumerate}
These two properties will then be used to prove that the \satcondition
condition and the reducibility condition which will allow us to apply
\reflem{lem:ar-method}.

\begin{definition}[Closure]
  \label{def:closure}
  Let $G=(V,E)$ be a graph and let $U \subseteq V$. We say that $U$ is
  \emph{closed} if $U = \Desc{U}$ and there are no 2\nobreakdash-,
  3\nobreakdash-, 4\nobreakdash-hops or lassos \wrt $U$. A
  \emph{closure} of $U$ is any minimal closed set that contains $U$.
\end{definition}

\begin{proposition}
  \label{clm:closure-single}
  Every set of vertices has a unique closure.
\end{proposition}
\begin{proof}
  Suppose there are two distinct closures $W_1$ and $W_2$ of $U$. As
  $W_1$ as well as $W_2$ contains $U$ it holds that
  $U \subseteq W_\cap = W_1 \cap W_2$. However, by minimality,
  $W_\cap$ is not a closure of $U$ and hence either there is a
  2\nobreakdash-, 3\nobreakdash-, 4\nobreakdash-hop or lasso $Q$
  \wrt~$W_\cap$ or there is a vertex $v$ in $\Desc{W_\cap}$ that is
  not in~$W_\cap$.
  
  In the first case this implies that some subgraph of $Q$ is either a
  2\nobreakdash-, 3\nobreakdash-, 4\nobreakdash-hop or a lasso \wrt
  $W_1$ or~$W_2$ since a subgraph of a hop or lasso is again hop or a
  lasso. In the second case this implies that
  $W_1 \not\supseteq \Desc{W_1}$ or $W_2 \not\supseteq \Desc{W_2}$,
  since $v \in \Desc{W_1} \cap \Desc{W_2}$.  Either way, this
  contradicts the assumption that both $W_1$ and $W_2$ are closed.
\end{proof}

In light of Proposition~\ref{clm:closure-single}, for a set $U$ we
write $\closure{U}$ to denote the unique closure of $U$. In order to
establish a bound on the size of the closure the algorithmic, but
otherwise equivalent, description in Algorithm~\ref{alg:closure} will
be useful.

\begin{algorithm}[t]
  \caption{A procedure to obtain the closure of a given set $U$.}
  \label{alg:closure}
  \begin{algorithmic}[1]
    \Procedure{Closure}{$U$}
    \State $W_0 \gets \Desc{U}$
    \State $i \gets 0$
    \While{exists 2\nobreakdash-, 3\nobreakdash-, 4\nobreakdash-hop or lasso $\hopexample_{i+1}$ \wrt $\Wset_{i}$}
    \State $W_{i+1} \gets \Desc{\Wset_{i} \union
      V(\hopexample_{i+1})}$
    \State $i \gets i+1$
    \EndWhile
    \State $W_{\mathrm{end}} \gets W_i$
    \State \textbf{return} $W_{\mathrm{end}}$
    \EndProcedure
  \end{algorithmic}
\end{algorithm}

Clearly the set of vertices $W_{\mathrm{end}}$ returned by
Algorithm~\ref{alg:closure} is closed and minimal by construction:
only adding part of a hop or a lasso results in a smaller hop. Hence
$\closure{U} = W_{\mathrm{end}}$ by
Proposition~\ref{clm:closure-single}.

The main property we use of our notion
of closure is that the neighbourhood of a closed set~$W$ is very
structured: since there are no $2$\nobreakdash-hops with respect to
$W$, every vertex in $\nbhstd{}{W}\setminus W$ has a
single neighbour in~$W$, and since there are no $3$\nobreakdash-hops
with respect to $W$ the neighbourhood of $W$ is an independent
set. The absence of longer $4$\nobreakdash-hops and lassos imply
similar, more technical properties for sets of vertices that are
connected to~$W$ via short paths.
We leverage this structure to argue that after removing a closed
set~$W$ from a sparse graph, all small vertex sets have a proper
$3$-colouring with the additional property that the neighbours of each
vertex at distance~1 from~$W$ are coloured with a single colour.

\begin{lemma}\label{lem:colouring-contraction}
  Suppose that $G=(V, E)$ is $(\ell,
  1/3\Delta)$\nobreakdash-sparse. Let $U, W\subseteq V$ be vertex sets
  such that $\setsize{U}\leq \ell$, $W$ is closed, and every vertex in
  $\nbhstd{U\setminus W}{W}$ has degree at most $\Delta$ in
  $G[U\setminus W]$. Then there exists a proper
  $3$\nobreakdash-colouring of the subgraph
  ${G[U\setminus (W \union \nbhstd{U}{W})]}$ such that for every
  $u \in \nbhstd{U\setminus W}{W}$, the set $\nbhstd{U\setminus W}{u}$
  is monochromatic.
\end{lemma}

\begin{proof}
  We start the proof by establishing the following two
  properties.
  \begin{enumerate}
  \item First, since there are no lassos in~$U$ \wrt $W$, for every
    vertex~$u\in \nbhstd{U\setminus W}{W}$, it holds that the
    neighbours~$\nbhstd{U\setminus W}{u}$ of $u$ are an independent
    set. Hence the
    graph~${\starg_u = {G[\{u\} \union \nbhstd{U\setminus W}{u}]}}$ is
    a star with center $u$ and leaves $\nbhstd{U\setminus W}{u}$.

  \item Second, as there are no 3\nobreakdash-, 4\nobreakdash-hops
    in~$U$ \wrt $W$, the graphs
    $\setdescr{\starg_u}{u\in \nbhstd{U\setminus W}{W}}$ as just
    defined are pairwise vertex disjoint.
  \end{enumerate}
  See \reffig{fig:stars} for an illustration.
  
  \begin{figure}
    \centering
    \includegraphics{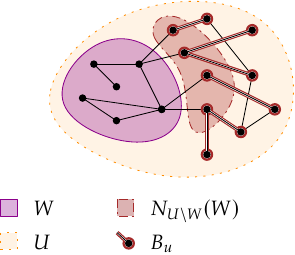}
    \caption{A depiction of the graphs $\starg_u$ as defined in the
      proof of \reflem{lem:colouring-contraction}.}
    \label{fig:stars}
  \end{figure}
  
  Let $G'$ be the graph obtained from
  $G[U\setminus W]$ by contracting each star $B_u$ to
  a vertex $\widetilde u$.
  Note that from a $3$\nobreakdash-colouring $\widetilde\chi$ of~$G'$
  we can define a $3$-colouring $\chi$ of
  ${G[U\setminus (W \union \nbhstd{U}{W})]}$ with the desired
  property: we claim that the 3-colouring
  $\chi \colon U\setminus (W \union \nbhstd{U}{W}) \rightarrow
  \set{1,2,3}$ defined by
  \begin{equation}
    \chi(w) =
    \begin{cases}
      \widetilde\chi(\widetilde u)
      &\text{if $w$ is a leaf of a star $B_u$,}\\
      \widetilde\chi(w)
      &\text{otherwise}
    \end{cases}
  \end{equation}
  is a proper 3-colouring of
  ${G[U\setminus (W \union \nbhstd{U}{W})]}$.  Note that since the
  graphs $\setdescr{B_u}{u\in \nbhstd{U\setminus W}{W}}$ are pairwise
  disjoint, the colouring $\chi$ assigns precisely one colour per
  vertex and, as each of the induced graphs~$B_u$ is a star, the
  colouring is proper. Finally, for all vertices
  $u \in \nbhstd{U\setminus W}{W}$ the set $\nbhstd{U\setminus W}{u}$
  is monochromatic by construction.
  
  It remains to argue that $G'$ is $3$\nobreakdash-colourable. Here we
  want to use the extreme sparsity of $G$: we enforce such extreme
  sparsity on $G$ that even though $G'$ is obtained from $G$ by
  contractions and hence has higher edge density than $G$, the graph
  $G'$ is still $(\setsize{V(G')}, 1/3)$\nobreakdash-sparse. Thus by
  \reflem{lem:sparse-colourable} it follows that~$G'$ is
  $3$\nobreakdash-colourable.

  Let us verify that~$G'$ is indeed
  $(\setsize{V(G')}, 1/3)$\nobreakdash-sparse, that is, we need to
  argue that every subset $T'\subseteq V(G')$ satisfies
  $|E(T')| \leq (1+1/3)|T'|$.
  Fix such a subset~$T'\subseteq V(G')$ and let $T$ be the preimage of
  $T'$ in the contraction. We estimate $\setsize{E(T')}$ in terms of
  $\setsize{E(T)}$.
  Let $\set{u_1, \ldots, u_t} = T\cap \nbhstd{U\setminus W}{W}$.
  Observe that in $T'$ each star $B_{u_i}$ is contracted to a vertex
  $\widetilde{u}_i$.  Let
  $s = \sum_{i=1}^t (\setsize{V(B_{u_i})} - 1)$. It holds that
  \begin{equation}
    \label{eq:starsizes}
    s\leq (\Delta-1)t \leq (\Delta-1)\setsize{T'} \eqcomma
  \end{equation}
  where we use the assumption that the degree of every vertex $u_i$ in
  $G[U\setminus W]$ is bounded by $\Delta$.  Furthermore, it holds
  that $\setsize{T} = \setsize{T'} + s$ and, because all the edges in
  the stars $\starg_{u_i}$ are contracted, that
  $\setsize{E(T)} \geq \setsize{E(T')} + s$.
  
  By assumption, $G$ is $(\ell, 1/3\Delta)$\nobreakdash-sparse and
  $\setsize{T} \leq \setsize{U} \leq \ell$. This implies that
  $\setsize{E(T)} \leq (1+1/3\Delta)|T|$ which, if combined with the above,
  gives $\setsize{E(T')} +s \leq (1+1/3\Delta)(\setsize{T'}+s)$. Using
  \refeq{eq:starsizes} we may conclude that
  \begin{equation}
    \label{eq:bottleneck}
    \setsize{E(T')}
    \leq
    \left(1+\frac{1}{3\Delta}\right)\setsize{T'} +
    \frac{s}{3\Delta}
    \leq
    \left(1+\frac{1}{3\Delta}\right)\setsize{T'} +
    \frac{(\Delta-1)\setsize{T'}}{3\Delta}
    =
    \left(1+\frac{1}{3}\right) \setsize{T'}
    \eqperiod
  \end{equation}
  Therefore, $G'$ is $(\setsize{G'}, 1/3)$\nobreakdash-sparse and
  thus, by \reflem{lem:sparse-colourable}, it follows that $G'$ is
  $3$\nobreakdash-colourable.
\end{proof}

The second property of the closure that we need is that we have some
control on how large it is, which is necessary to show that the
\satcondition condition of \reflem{lem:ar-method} holds.  We establish
an upper bound on the size of the closure by proving that for sparse
graphs the closure of a not too large set $U$ is the set of
descendants of a set $\Dset$ that is not much larger than $U$, as
stated in the lemma below.
Combining this with \reflem{lem:sparse-colourable}, it follows that in
order to establish that the set $\closure{U}$ is $3$-colourable---and
thus $\Col{G[\closure{U}]}{3}$ is satisfiable---it is sufficient to
prove an upper bound on the size of the set of descendants of $\Dset$.
Jumping ahead a bit, we will be able to establish such an upper bound
in the next section by choosing an appropriate ordering of the
vertices.

\begin{lemma}\label{lem:closure-lemma}
  Suppose that $G=(V, E)$ has a linear order on $V$ and is
  $(\ell, 1/3\sparseparam)$\nobreakdash-sparse for
  $\sparseparam\geq 2$.
  Let~$U\subseteq V$ be a set of size
  $\setsize{U} \leq \ell/25\sparseparam$ such that any decreasing path
  in $G[V\setminus U]$ has at most $\sparseparam$ vertices.
  Then there exists a set $\Dset\subseteq V$ such that
  $\Dset\supseteq U$, $\setsize{\Dset} \leq 25 \setsize{U}$ and
  ${\closure{U}} = \Desc{\Dset}$.
\end{lemma}

\begin{proof}
  Recall from Algorithm~\ref{alg:closure} that the closure of a set
  $U\subseteq V$ can be defined to be the final set in a sequence
  $(\Wset_0, \Wset_1, \ldots, \Wset_{\mathrm{end}})$, where
  $\Wset_0 =\Desc{U}$ and $\Wset_{i+1}$ is obtained from~$\Wset_{i}$
  by appending a 2\nobreakdash-, 3\nobreakdash-, 4\nobreakdash-hop or
  a lasso \wrt $\Wset_{i}$ and then taking the descendants of the
  resulting set of vertices. A key observation is that adding such
  hops or lassos to $\Wset_i$ adds more edges to the induced graph
  $G[\Wset_i]$ than vertices, thus increasing the edge density. As the
  graph~$G$ is locally sparse we can conclude that the sequence
  $(\Wset_0, \Wset_1, \ldots, \Wset_{\mathrm{end}})$ needs to be
  rather short, which allows us to argue the size upper bound on
  $\Dset$.

  In the following,
  for each $\Wset_i$ we identify a vertex set~$U_i$ such that the edge
  density of the graph~$G[U_i]$ increases with $i$. The idea is as
  follows. Since the sets $U_i$ grow very slowly and thus the local
  sparsity always applies, we will be able to conclude that the number
  of iterations in the construction of $\closure{U}$ is bounded.  As
  the vertices in $\Wset_{i}$ are the descendants of the set $U_i$, it
  holds that the vertices in $\Wset_{\mathrm{end}}$ are the
  descendants of a set which is not much larger than the initial set
  $U$, whereby the lemma follows.
  
  Let us now implement this plan. We define $U_i$ inductively as
  follows. Let $U_0 = U$ and let~$\hopexample_i$ be the hop or lasso
  added to $\Wset_{i-1}$ at iteration~$i\geq 1$. If we denote by $u$
  and $v$ the endpoints of $\hopexample_i$ (where we could have
  $u = v$) and let $\graphpath_u$ and $\graphpath_v$ be two shortest
  decreasing paths from $U_{i-1}$ to $u$ and~$v$, respectively, then
  it holds that
  $U_i = U_{i-1} \union V(P_u \union P_v \union \hopexample_i)$.  See
  \reffig{fig:s_i} for an illustration.

  For our definition of $U_i$ to be meaningful, we need to establish
  that the paths $\graphpath_u$ and $\graphpath_v$ always exist.

  \begin{figure}
    \centering
    \includegraphics{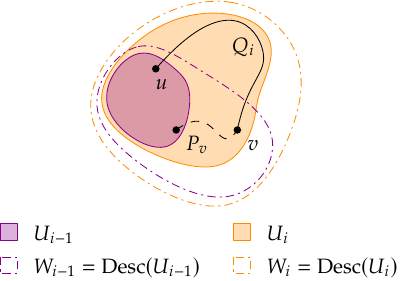}
    \caption{A depiction of the construction of $U_i$ as defined in the
      proof of \reflem{lem:closure-lemma}.}
    \label{fig:s_i}
  \end{figure}
  \begin{claim}
    \label{claim:si-well-defined}
    For every vertex~$v$ in~$\Wset_i$, there exists
    a decreasing path in~$\Wset_i$ from some vertex in~$U_i$ to $v$. 
  \end{claim}
  \begin{proof}
    The proof is by induction on $i$. The base case $i = 0$ holds
    because $\Wset_0 = \Desc{U_0}$.  For the induction step, suppose
    that the claim holds for $i-1$. By definition, the vertices
    in~$\Wset_i\setminus \Wset_{i-1}$ are descendants of a vertex
    in~$\hopexample_i$, and all vertices in~$\hopexample_i$ are
    contained in~$U_i$. The claim follows.
  \end{proof}
  Next, we show that~$\setsize{U_i}$ grows slowly with $i$ and that
  the edge density $\setsize{E(U_i)}/\setsize{U_i}$ exceeds the
  sparsity threshold $(1+ 1/3\sparseparam)$ after a small number of
  iterations.

  \begin{claim}
    \label{claim:si-bounded} 
    For $U_i$ as commented above it holds
    that~$\setsize{U_i \setminus U_{i-1}}\leq 2\sparseparam +
    \setsize{V(\hopexample_i)} - 4$ and
    $\setsize{E(U_i)} \geq \setsize{E(U_{i-1})} + \setsize{U_i
      \setminus U_{i-1}} + 1$.
  \end{claim}
  \begin{proof}
    Let~$F$ be the graph defined by the union of the edges in~$P_u,
    P_v$ and~$Q_i$.
    Since the paths $\graphpath_u$ and $\graphpath_v$ are decreasing
    paths, according to the statement of \reflem{lem:closure-lemma},
    they contain at most $\sparseparam$ vertices each, so
    ${3 \leq \setsize{V(F)} \leq 2\sparseparam +
      \setsize{V(\hopexample_i)} - 2}$. Moreover, the endpoints of~$F$
    are contained in~$U_{i-1}$ and all other vertices in~$F$ are
    outside of $U_{i-1}$. By our choice of $P_u$ and $P_v$ there are
    two cases, depending on whether $F$ contains a cycle or not.
    
    \textbf{Case 1:} If there is no cycle in~$F$, then
    $\setsize{V(F)\intersection U_{i-1}} = 2$ so
    $\setsize{U_i\setminus U_{i-1}}= \setsize{V(F)} - 2$. Moreover
    $\lvert E(F) \rvert \geq \setsize{V(F)} - 1$ since $F$ is
    connected.

    \textbf{Case 2:} If $F$ contains a cycle, then
    $\setsize{V(F)\intersection U_{i-1}} = 1$, hence
    $\setsize{U_i\setminus U_{i-1}}= \setsize{V(F)}-1$. In addition,
    $\lvert E(F) \rvert \geq \setsize{V(F)}$ since $F$ is connected
    and contains a cycle. Moreover, it holds that
    ${\setsize{V(F)} \leq 2\sparseparam +
      \setsize{V(\hopexample_i)} - 3}$.

    Since in both cases the number of added vertices is
    bounded~${\setsize{U_i\setminus U_{i-1}} \leq 2\sparseparam +
      \setsize{V(\hopexample_i)} - 4}$ and the number of edges in the
    subgraph induced by~$U_i$ can be lower
    bounded~$\setsize{E(U_i)} \geq \setsize{E(U_{i-1})} + \setsize{U_i
      \setminus U_{i-1}} + 1$ the statement follows.
  \end{proof}

  Recall that we want to show that the edge density of the induced
  subgraph~$G[U_i]$ increases with $i$. Since $G$ is sparse it thus
  follows that the number of iterations in the construction of
  $\closure{U}$ is bounded.

  Let $s = \setsize{U}$. Towards contradiction, suppose
  that~$i \geq 5s + 1$. Note that~$|V(Q_i)|\leq 5$, so by
  \refclaim{claim:si-bounded} we have
  $\setsize{U_i\setminus U_{i-1}} \leq 2\sparseparam + 1$. This
  implies that
   \begin{align}
     {\setsize{U_{5s+1}}  = s +
     \sum_{i=1}^{5s + 1}
     \setsize{U_{i}\setminus U_{i-1}}
     \leq s + (5s + 1) (2\sparseparam + 1) 
         < 16\sparseparam\setsize{U} <
           \ell} \eqcomma
   \end{align}
   and therefore, since $G$ is
   $(\ell, 1/3\sparseparam)$\nobreakdash-sparse, it holds that
   ${ \setsize{E(U_{5s + 1})} }/{ \setsize{U_{5s+1}} } \leq 1 +
   1/3\sparseparam$.  However, by \refclaim{claim:si-bounded}, we have
   that
   \begin{subequations}\label{eq:density}
   \begin{align}
     \frac{
     \setsize{E(U_{5s + 1})}
     }{
     \setsize{U_{5s+1}}
     }
     &= 
       \frac{
       \setsize{E(U)}+
       \sum_{i=1}^{5s + 1}
       \setsize{E(U_{i})
       \setminus
       E(U_{i-1})}
       }{
       s+
       \sum_{i=1}^{5s + 1}
       \setsize{U_{i}\setminus U_{i-1}}
       } \\ \label{eq:size-1}
     &\geq
       \frac{
       \setsize{E(U)}+
       \sum_{i=1}^{5s + 1}
       \left(
       \setsize{U_{i}
       \setminus
       U_{i-1}
       }+
       1
       \right)
       }{
       s+
       \sum_{i=1}^{5s + 1}
       \setsize{U_{i}\setminus U_{i-1}}
       }\\ \label{eq:size-2}
     &\geq
       \frac{
         0 + \sum_{i=1}^{5s+1}(2\sparseparam+2)
       }{
         s + \sum_{i=1}^{5s+1}(2\sparseparam+1)
       }
     \\ \label{eq:size-3}
     &>
     1 +
     \frac{1}{3\sparseparam}\eqcomma
   \end{align}
 \end{subequations}
 where for \refeq{eq:size-1} we use that
 ${\setsize{E(U_i)} \geq \setsize{E(U_{i-1})} + \setsize{U_i \setminus
     U_{i-1}} + 1}$, and then for \refeq{eq:size-2}, we observe that
 the fraction decreases as
 $ \sum_{i=1}^{5\setsize{U} + 1} \setsize{U_{i}\setminus U_{i-1}}$
 increases and thus, along with the bound
 $\setsize{U_i\setminus U_{i-1}} \leq 2\sparseparam + 1$, we obtain
 the claimed inequality. For the final inequality \refeq{eq:size-3} we
 use $\sparseparam\geq 2$. This contradicts the assumption that $G$ is
 $(\ell,1/3\sparseparam)$-sparse and it hence follows
 that~${i \leq 5\setsize{U}}$.

 Let $i_{\mathrm{end}}$ be the last iteration and let
 $\Dset = U \union \bigcup_{j \leq i_{\mathrm{end}}} V(\hopexample_j)$.
 We claim that $\Desc{\Dset} = \closure{U}$ and
 $\setsize{\Dset} \leq 25 \setsize{U}$.  Indeed, at iteration $i$ in the
 construction, $\Wset_{i}$ is the set of descendants of
 $U \union \bigcup_{j \leq i} V(\hopexample_j)$, and the hop or
 lasso~$\hopexample_i$ added to $\Wset_{i-1}$ contains at most $4$
 vertices not already in $\Wset_i$.  Therefore, $\Desc{\Dset} = \closure{U}$
 and since $i_{\mathrm{end}} \leq 5\setsize{U}$ it follows that
 $\setsize{\Dset} \leq \setsize{U} + 4\cdot 5\setsize{U} \leq
 25\setsize{U}$.
\end{proof}

\section{A Lower Bound for 3-Colourability on Sparse Random Graphs} 
\label{sec:3-colouring-hard-merged}
We are now ready to prove a linear degree lower bound for polynomial
calculus refutations of the claim that sparse random graphs are
$3$-colourable. In fact, we prove something slightly stronger: we show
that if a graph is locally very sparse and has only few vertices of
high degree, then it is hard for polynomial calculus to refute the
claim that the graph is $3$-\nobreakdash colourable.

\begin{theorem}[Main theorem]
  \label{th:3-colouring-hard-sparse}
  Let $c,\Delta,k$ and $\ell$ be integers such that $c > k\geq 3$ and
  let $G=(V,E)$ be a $c$-colourable,
  $(\ell, 1/3\Delta)$\nobreakdash-sparse graph.  For any
  set~$\badset_\Delta$ such that $G[V\setminus\badset_\Delta]$ has
  maximum degree at most $\Delta$, any polynomial calculus
  refutation of $\Col{G}{k}$ over any field requires degree
  $\ell/50\Delta^{c-1} - \setsize{\badset_\Delta}$.
\end{theorem}

We defer the proof of \refthm{th:3-colouring-hard-sparse} to
\refsec{sec:proof-main}, and first show how our results follow from
this theorem. Intuitively, the above theorem holds because sparse
graphs are locally $3$-colourable (see \reflem{lem:sparse-colourable})
and hence the colouring formula defined over sparse graphs is locally
satisfiable. Before diving into the proof of
\refthm{th:3-colouring-hard-sparse}, let us see how we can use
\reflem{lem:razborov-sparse} to obtain degree lower bounds for random
graphs. Note that the condition on the graph $G$ in
\refthm{th:3-colouring-hard-sparse} is not inherently random. However,
to the best of our knowledge, there are no explicit constructions of
graphs that are this sparse. In fact, it is stated as an open problem
in the survey by Hoory et al.~\cite[Open problem
10.8]{HLW06ExpanderGraphs} to explicitly construct such graphs.

Recall from \reflem{lem:chromatic-number-random} that if $d \geq 6$,
then graphs sampled from the \erdosrenyi random graph distribution
$\gndn$ or the random $d$-regular graph distribution $\gnd$ are \aas
not $3$-colourable. In light of this, the below statements are
interesting in the parameter regime $d \geq 6$.

\begin{corollary}[Colouring lower bound for random regular graphs]
  \label{th:3-colouring-hard}
  There exists an absolute constant $C$ such that for positive
  integers $n$ and $d \geq 2$ satisfying $6 d^3 \leq \log n$ the
  following holds.
  If $G$ is a graph sampled from $\gnd$ and $k\geq 3$ is an integer,
  then \aas every polynomial calculus refutation of $\Col{G}{k}$ over
  any field requires degree $d^{-Cd} \cdot n$.
\end{corollary}
\begin{proof}
  Let $\Delta = d$ and $\badset_\Delta = \emptyset$.  By
  \reflem{lem:chromatic-number-random} we have that \aas~$G$ is
  $c$\nobreakdash-colourable for $c \leq 2 d/\log{d}$, and by
  \reflem{lem:razborov-sparse} that \aas~$G$ is
  $(\ell, 1/3d)$\nobreakdash-sparse for~${\ell = (8d)^{-6d}\nvars}$.
  Note that to apply \reflem{lem:razborov-sparse}, we use
  $\epsilon = 1/3d$ and $\delta = 1/2$. We can now apply
  \refth{th:3-colouring-hard-sparse} and conclude that any polynomial
  calculus refutation of $\Col{G}{k}$, over any field, requires degree
  $\ell/50\Delta^{c-1} - \setsize{\badset_\Delta} =
  (8d)^{-6d}\nvars/50d^{2d/\log d-1} \geq d^{-Cd} \cdot n$, using that
  $C$ is a large enough constant.
\end{proof}

To prove the result for \erdosrenyi random graphs we need the
additional property that \aas there exists a small set
$\badset_\Delta$ of vertices such that $G[V\setminus\badset_\Delta]$
has maximum degree at most $\Delta$.

\begin{lemma}\label{lem:TDelta-small-v2}
  Let~$G=(V, E)$ be a graph sampled from~$\gndn$ where
  $d=\bigoh{\log n}$.  If $\Delta\geq d$ is such that
  $(\Delta/ed)^\Delta = \littleoh{n}$, then \aas there exists a set
  $\badset_\Delta$ of size at most $(ed/\Delta)^\Delta \cdot 2en$ such
  that the maximum degree in~$G[V\setminus \badset_\Delta]$ is at
  most~$\Delta-1$.
\end{lemma}

The proof of \reflem{lem:TDelta-small-v2} is mostly a standard
calculation and we present it in Appendix~\ref{sec:TDelta-small} for
completeness. We are now ready to prove our lower bound for
\erdosrenyi random graphs.

\begin{corollary}[Colouring lower bound for \erdosrenyi random graphs]
  \label{th:3-colouring-hard-gnp}
  There exists an absolute constant $C$ such that for $n \in \N^+$ and
  $d \in \R^+$ satisfying that $d > 1$ and
  $ d^{5} = \littleoh{\log n}$ the following holds.
  If $G$ is a graph sampled from $\gndn$ and $k\geq 3$ is an integer,
  then \aas every polynomial calculus refutation of $\Col{G}{k}$ over
  any field requires degree $d^{-Cd^5} \cdot n$.
\end{corollary}

\begin{proof}
  Fix $\Delta = (5 d)^5$.  By \reflem{lem:chromatic-number-random} we
  have that \aas~$G$ is $c$-colourable for $c \leq 2 d/\log{d}$, and
  by \reflem{lem:razborov-sparse} it holds \aas that~$G$ is
  $(\ell, \epsilon)$\nobreakdash-sparse
  for~${\ell = (4d)^{-4(5d)^5}\nvars}$ and $\epsilon = 1/3\Delta$.
  Note that to apply \reflem{lem:razborov-sparse}, we can use
  $\delta = 1/4$, so that
  $(1+\epsilon)(1+\delta)/\epsilon \leq 4\Delta$.

  Let~$\badset_\Delta\subseteq V$ be a minimum size set such that
  $G[V\setminus\badset_\Delta]$ has maximum degree at most
  $\Delta$. By \reflem{lem:TDelta-small-v2}, we have that
  \begin{equation}
    \setsize{\badset_\Delta}
    \leq
    (ed/\Delta)^\Delta \cdot 2en
    =
    \frac{ 2en}{(5^5d^4/e)^{(5d)^5}}
    <
    \frac{ n}{2 \cdot (5d)^{4(5d)^5}}
    \eqcomma
  \end{equation}
  where the last inequality follows since
  $(5/e)^{(5d)^5} \geq (5/e)^{5^5} > 4e$.

  We can now apply \refth{th:3-colouring-hard-sparse} and conclude
  that any polynomial calculus refutation of $\Col{G}{k}$, over any
  field, requires degree
  \begin{subequations}
    \begin{align}
      \frac{\ell}{50\Delta^{c-1}} -
      \setsize{\badset_\Delta} 
      &\geq
        \frac{ n }
        {
        50
        \cdot
        (5 d)^{10d/\log d}
        \cdot
        (4d)^{4(5d)^5}
        }
        -
        \frac{n}{2\cdot (5d)^{4(5d)^5}}
      \\
      &\geq
        \frac{n }{ (5d)^{4(5d)^5} }
        -
        \frac{n}{2\cdot (5d)^{4(5d)^5}}
      \\
      &\geq
        d^{-Cd^5}
        \cdot
        n
        \eqcomma
    \end{align}
  \end{subequations}
  where for the second inequality we use that
  $50 \cdot (5d)^{10d/\log d} < 2^{10} \cdot (10)^{10d} \leq
  2^{(5d)^5}$ and for the last inequality we use that $C$ is a large
  enough constant.
\end{proof} 

\subsection{Proof of Main Theorem}
\label{sec:proof-main}

Fix $\badset_\Delta$ such that $G[V \setminus \badset_\Delta]$ has
maximum degree at most $\Delta$ and let
$X = \Set{x_{v, i}\mid v\in V(G),\; i \in [k]}$.

To prove \refth{th:3-colouring-hard-sparse}, our goal is to define a
\acceptable{\Col{G}{k}} $S$ that maps monomials in the polynomial
ring~$\F[X]$ to subsets of $\Col{G}{k}$ such that
\reflem{lem:ar-method} holds.  For brevity, given a set $U\subseteq V$
we denote the ideal $\langle \Col{G[U]}{k}\rangle$ by
$\langle U \rangle$ and refer to the polynomials in~$\Col{G[U]}{k}$ as
the generators of~$\langle U\rangle$.

Given a monomial $m$, we let $V(m)$ denote the set of vertices
mentioned by the variables in~$m$.  Moreover, given a linear order
$\vertexorder$ on the vertices, we say an admissible order
$\monomialorder$ of the monomials over the variables $X$ of
$\Col{G}{k}$ \emph{respects $\vertexorder$} if for any colours
$i, j \in [k]$ it holds that $x_{u,i} \monomialorder x_{v,j}$ whenever
$u \vertexorder v$.

The main technical lemma we need in order to prove
\refth{th:3-colouring-hard-sparse} is the reducibility lemma below,
from which the reducibility condition of \reflem{lem:ar-method} will
follow.  The reducibility lemma implies, in particular, that reducing
a monomial $m$ modulo $\langle W\rangle$ for any closed set
$W\supseteq V(m) \cup \badset_\Delta$ is the same as reducing modulo
$\langle U\rangle$ for any superset $U\supseteq W$ that is not too
large.

\begin{lemma}[Reducibility lemma]
  \label{lem:local-reduction-strong}
  Let $G=(V, E)$ be a $(\ell, 1/3\Delta)$\nobreakdash-sparse graph
  with a linear order $\vertexorder$ on $V$ and consider an admissible
  order that respects~$\vertexorder$. If the vertex sets
  $W\subseteq U$ satisfy hat $W$ is closed, that the size of $U$ is
  $\setsize{U} \leq \ell$, and that every vertex in
  $\nbhstd{U\setminus W}{W}$ has degree at most $\Delta$ in
  $G[U\setminus W]$, then for every monomial $\pcmonm$ such that
  $V(m) \subseteq W$, it holds that $\pcmonm$ is reducible modulo
  $\langle U\rangle$ if and only if $\pcmonm$ is reducible modulo
  $\langle W \rangle$.
\end{lemma}

We postpone the proof of this lemma to the end of this section.
In order to ensure that we are always reducing a monomial $m$ by $\langle W \rangle$
for some closed set $W$ that contains $V(m)$ and is such that 
every vertex in $\nbhstd{U\setminus W}{W}$ has degree at most $\Delta$ in $G[U \setminus W]$
for all $U\supseteq W$, we define 
the closure of a monomial to include the set $\badset_\Delta$.

\begin{definition}[Monomial closure]
  \label{def:mon-closure-Gnp}
  The \emph{monomial closure} of a monomial $m$, denoted by $\closure[\Delta]{m}$, is the
  vertex set $\closure{V(m)\cup \badset_\Delta}$.
\end{definition}

Looking ahead, we will prove \refth{th:3-colouring-hard-sparse} by
showing that the map $S$ defined by mapping monomials $m$ to
$\Col{G[\closure[\Delta]{\pcmonm}]}{k} \subseteq \Col{G}{k}$ is a
\acceptable{\Col{G}{k}} and appealing to \reflem{lem:ar-method}. To
establish the \satcondition condition in \reflem{lem:ar-method}, we
must prove that $\Col{G[\closure[\Delta]{m}]}{k}$ is satisfiable
whenever $m$ is of low degree. Since $G$ is sparse, by
\reflem{lem:sparse-colourable}, it suffices to show that the monomial
closure of $m$ is not too large. This, in turn, will follow from the
next lemma. \reflem{lem:size-lemma} is an almost direct consequence of
\reflem{lem:closure-lemma} and states that under suitable technical
assumptions, the size of the monomial closure of~$m$ is closely
related to the degree of~$m$ and the size of $\badset_\Delta$.

\begin{lemma}[\Satlemma]
  \label{lem:size-lemma}
  Let $\sparseparam, \Delta,\ell \in \N^+$ such that $a \geq 2$ and
  let $G=(V,E)$ be a $(\ell, 1/3\sparseparam)$\nobreakdash-sparse
  graph with a linear order on $V$. Let $\badset_\Delta \subseteq V$
  be such that $G[V\setminus \badset_\Delta]$ has maximum degree at
  most~$\Delta$, $\Desc{\badset_\Delta} \subseteq \badset_\Delta$, and
  that any decreasing path in $G[V\setminus \badset_\Delta]$ has at
  most $\sparseparam$ vertices.  Then it holds for any monomial $m$
  that if
  $\mdegreeof{m} + \setsize{\badset_\Delta} \le \ell/25\sparseparam$,
  then
  $\setsize{\closure[\Delta]{m}} \leq 50 \Delta^{\sparseparam - 1}
  \cdot (\mdegreeof{m} + \setsize{\badset_\Delta})$.
\end{lemma}
\begin{proof}
  Let $U = V(m) \union T_\Delta$ and note that
  $\closure[\Delta]{m} = \closure{U}$.
  Note that
  $\setsize{U} =
  \setsize{ V(m) \union T_\Delta}
  \leq
  \mdegreeof{m} + \setsize{\badset_\Delta}
  \leq \ell /25\sparseparam$
  and that any
  decreasing path in $G[V\setminus U]$ has at most $\sparseparam$
  vertices. We can therefore apply \reflem{lem:closure-lemma} to
  deduce that there exists a set $\Dset\subseteq V$ such that
  $U\subseteq \Dset$, $\closure{U} = \Desc{\Dset}$ and
  $\setsize{\Dset} \leq 25 \setsize{U}$.
    
  Note that since $\Dset\supseteq \badset_\Delta$ and all the
  descendants of vertices in $\badset_\Delta$ are in~$\badset_\Delta$,
  we have that
  $\Desc{\Dset} = \left(\Desc{\Dset\setminus \badset_\Delta}\setminus
    \badset_\Delta \right) \cup \badset_\Delta$.  Moreover, since any
  vertex $v \in \Dset\setminus \badset_\Delta$ has degree at most
  $\Delta$ in $G[V\setminus \badset_\Delta]$, and since any decreasing
  path in $G[V\setminus \badset_\Delta]$ has at most $\sparseparam$
  vertices, it follows that $v$ has at most
  $2 \Delta^{\sparseparam - 1}$ descendants in
  $V\setminus \badset_\Delta$.  We thus have the upper bound
  $\setsize{\left(\Desc{\Dset\setminus \badset_\Delta}\setminus
      \badset_\Delta \right)} \leq 2 \Delta^{\sparseparam - 1} \cdot
  \setsize{\Dset\setminus \badset_\Delta} $ from which we conclude
  that
  \begin{equation*}
    \setsize{\Desc{\Dset}}
    \leq
    2 \Delta^{\sparseparam - 1}
    \cdot
    \setsize{\Dset\setminus \badset_\Delta}
    +
    \setsize{\badset_\Delta}
    \leq
    2 \Delta^{\sparseparam - 1}
    \cdot
    \setsize{\Dset}
    \leq
    50
    \Delta^{\sparseparam - 1}
    \cdot
    (\mdegreeof{m} + \setsize{\badset_\Delta})
    \eqcomma
  \end{equation*}
  as claimed in the lemma.
\end{proof}

We are now ready to prove our main theorem, which we restate here for
convenience.

\begin{restatablethm}{\ref{th:3-colouring-hard-sparse}}[Main theorem, restated]
  Let $c,\Delta,k$ and $\ell$ be integers such that $c > k\geq 3$ and
  let $G=(V,E)$ be a $c$-colourable,
  $(\ell, 1/3\Delta)$\nobreakdash-sparse graph.  For any set
  set~$\badset_\Delta$ such that $G[V\setminus\badset_\Delta]$ has
  maximum degree at most $\Delta$, then any polynomial calculus
  refutation of $\Col{G}{k}$ over any field requires degree
  $\ell/50\Delta^{c-1} - \setsize{\badset_\Delta}$.
\end{restatablethm}

\begin{proof}[Proof of \refth{th:3-colouring-hard-sparse}]
  We start by defining a linear order on $V$.  Let
  $\chi_c\colon V\setminus\badset_\Delta \to [c]$ be a proper
  $c$\nobreakdash-colouring of $G[V\setminus\badset_\Delta]$.  We let
  the order~$\vertexorder$ on $V$ be any linear order that satisfies
  $u \vertexorder v$ whenever $u\in \badset_\Delta$ and
  $v\in V\setminus \badset_\Delta$ and whenever
  $u,v\in V\setminus \badset_\Delta$ and~$\chi_c(u) < \chi_c(v)$.
  Observe that any decreasing path in $G[V\setminus \badset_\Delta]$
  has at most $c$ vertices and that
  $\Desc{\badset_\Delta} \subseteq \badset_\Delta$.
  
  We can define an admissible ordering $\monomialorder$ of the
  monomials over the variables $\set{x_{v,i}}_{v\in V, i\in [k]}$ of
  $\Col{G}{k}$ that respects $\vertexorder$ as follows: for distinct
  vertices~$u,v$ let~$x_{u,i} \monomialorder x_{v,j}$
  whenever~$u \vertexorder v$ and for variables associated with the
  same vertex~$u$ let~$x_{u, i} \monomialorder x_{u,j}$
  whenever~$i < j$.
  With this order fixed we then obtain the admissible ordering on
  monomials by first ordering the monomials by degree and then
  lexicographically according to the ordering on the variables.

  To prove \refth{th:3-colouring-hard-sparse} we
  use~\reflem{lem:ar-method}.  For this, we show that the map
  $S\colon m \mapsto \Col{G[\closure[\Delta]{\pcmonm}]}{k}$ is a
  \acceptable{\Col{G}{k}} in the sense of \refdef{def:acceptable} and
  that it satisfies the \satcondition condition and the reducibility
  condition in \reflem{lem:ar-method} for
  $D= \ell / (50 \Delta^{\gcolourability-1}) -
  \setsize{\badset_\Delta}$.

  We start by proving, via the properties of the
  closure, that the map~$S$ is a \acceptable{\Col{G}{k}}.  We actually show
  something slightly stronger, namely that $S$ satisfies the following
  four properties, which are the same as those in
  \refdef{def:acceptable} except that we do not require that
  $m' \prec m$ in \refitem{it:property-Sxm-again} and
  \refitem{it:property-Sm-again}.
   \begin{enumerate}
   \item For all monomials $m$ and $m'$ such that $\vars{m} = \vars{m'}$, 
   it holds that $S(m) =  S(m')$.
     \label{it:property-multilinear-again}
   \item For every variable $x$ and for all monomials $m$ and $m'$, if
     $S(m')\subseteq \Smap{\pcmonm}$, then $S(xm')\subseteq
     S(xm)$. 
     \label{it:property-Sxm-again}
   \item For all monomials $m$ and $m'$, if
     $\vars{m'}\subseteq\Vars{\Smap{\pcmonm}}$, then $S(m')\subseteq \Smap{\pcmonm}$.
     \label{it:property-Sm-again}
   \item For all $p\in \pcsp$, it holds that $p\in S(m)$, where $m$ is
     the leading monomial in $p$.
     \label{it:property-axiom-again}
   \end{enumerate}  
   Item~\ref{it:property-multilinear-again} follows immediately from
   the definition of $S$, since the closure of a monomial~$m$ only
   depends on $V(m)$.
   Note that since $S(m) = \Col{G[\closure[\Delta]{\pcmonm}]}{k}$ it
   holds that $S(m' ) \subseteq S(m)$ if and only if
   $\closure[\Delta]{m'} \subseteq \closure[\Delta]{m}$ and,
   therefore, \refitem{it:property-Sxm-again} is equivalent to showing
   that if $\closure[\Delta]{m'}\subseteq \closure[\Delta]{m}$, then
   $\closure[\Delta]{xm'}\subseteq \closure[\Delta]{xm}$.  Recall that
   $\closure[\Delta]{m} = \closure{V(m) \cup \badset_\Delta}$ and
   hence, by minimality of closure, we obtain that
   \begin{subequations}
     \label{eq:contains}
     \begin{align}
       V(x) \union V(m') \union \badset_\Delta
       &\subseteq
         \closure{V(x)}
         \union
         \closure{V(m')
         \cup
         \badset_\Delta}
       &&[\text{since $U \subseteq \closure{U}$}] \\
       &\subseteq
         \closure{V(x)}
         \union
         \closure{V(m)
         \cup
         \badset_\Delta}
       &&[\text{as $\closure[\Delta]{m'}\subseteq
          \closure[\Delta]{m}$ by assumption}] \\
       &\subseteq
         \closure{
         V(xm) \union
         \badset_\Delta}  \eqcomma
     \end{align}
   \end{subequations}
   where the final equation relies on the fact that
   $\closure{A} \cup \closure{B} \subseteq \closure{A\cup B}$ which
   follows by minimality of the sets $\closure{A}$ and $\closure{B}$.
   This allows us to derive that
  \begin{subequations}
    \begin{align}
      \closure[\Delta]{xm'} 
      &= \closure{V(x) \union V(m') \union \badset_\Delta}
      &&[\text{by definition of monomial closure}] \\ 
      &\subseteq
        \closure{
        \closure{V(xm) \union \badset_\Delta}}
      &&[\text{by \refeq{eq:contains}}] \\
      &=  \closure{V(xm) \union \badset_\Delta}
      &&[\text{since closure is idempotent by minimality}] \\
      &=  \closure[\Delta]{xm} \eqcomma
      &&[\text{by definition of monomial closure}]
    \end{align}
  \end{subequations}
  and thus \refitem{it:property-Sxm-again} holds. Observe furthermore
  that if $\vars{m' } \subseteq \vars{S(m)}$ then
  $V(m') \subseteq \closure[\Delta]{m}$ by the definition of monomial
  closure and since $S(m) =
  \Col{G[\closure[\Delta]{\pcmonm}]}{k}$. Thus, using again the
  observation that $S(m' ) \subseteq S(m)$ if and only if
  $\closure[\Delta]{m'} \subseteq \closure[\Delta]{m}$, in order to
  conclude that \refitem{it:property-Sm-again} holds it suffices to
  show that if $V(m')\subseteq \closure[\Delta]{m}$, then
  $\closure[\Delta]{m'}\subseteq \closure[\Delta]{m}$.  By again using
  the minimality of closure and the fact that
  $U \subseteq \closure{U}$, we can conclude that
  \begin{align}
    {\closure[\Delta]{m'}
    =
    \closure{V(m') \union \badset_\Delta}
    \subseteq
    \closure{\closure[\Delta]{m} \union \badset_\Delta}
    =
    \closure[\Delta]{m}} \eqperiod
  \end{align}
  Finally, item \ref{it:property-axiom-again} follows easily from the definition of $S$.
  Indeed, if $\pcpolyp \in \Col{G}{k}$ is an edge axiom, say $x_{u, i}x_{v, i}$, 
  then it holds that
  $S(x_{u, i}x_{v, i}) = \Col{G[\closure{\set{u,v} \cup \badset_\Delta}]}{k} \ni p$;
  and
  if $\pcpolyp \in \Col{G}{k}$ is a vertex axiom 
  ($\sum_{i=1}^\kcolourcons x_{v, i} - 1$ or $x_{v, i}x_{v, i'}$) or a 
  Boolean axiom ($x_{v, i}^2 - x_{v, i}$)
  mentioning a vertex $v$ and $m$ is the leading monomial in $p$ 
  it holds that
  $ S(m) = \Col{G[\closure{\set{v} \cup \badset_\Delta}]}{k} \ni p$.
  Thus, the map $S$ is a \acceptable{\Col{G}{k}}.

  We now show that the \satcondition condition and the reducibility
  condition in \reflem{lem:ar-method} hold for the map~$S$ and for
  $D = \ell / (50 \Delta^{\gcolourability-1}) -
  \setsize{\badset_\Delta}$.  We can assume $D\geq 2$, since otherwise
  the theorem is trivially true.  Observe that the polynomials in
  $\Col{G}{k}$ have degree at most $D$.

  To see that the \satcondition condition holds, note that by
  \reflem{lem:size-lemma} every monomial $m$ of degree at most~$D$
  satisfies~$\setsize{\closure[\Delta]{m}} \leq 50 \Delta^{c-1} (D +
  \setsize{\badset_\Delta}) = \ell$, where we use that any decreasing
  path in $G[V\setminus \badset_\Delta]$ has at most $c$ vertices,
  that $\Desc{\badset_\Delta} \subseteq \badset_\Delta$ holds, and
  that $G[V\setminus \badset_\Delta]$ has maximum degree at most
  $\Delta$.  Since $G$ is $(\ell, 1/3\Delta)$\nobreakdash-sparse and
  $\setsize{\closure[\Delta]{m}} \leq \ell$, it follows from
  \reflem{lem:sparse-colourable} that the
  graph~$G[\closure[\Delta]{m}]$ is~$3$-colourable and so
  $\Col{G[\closure[\Delta]{m}]}{k}$ is satisfiable.

  To establish the reducibility condition, let $m$ and $m'$ be
  monomials of degree at most $D$ such that $S(m') \subseteq S(m)$,
  \ie such that $\closure[\Delta]{m'} \subseteq \closure[\Delta]{m}$.
  Note that as argued above, it holds that
  $\setsize{\closure[\Delta]{m}}\leq \ell$ since~$m$ has degree at
  most $D$.  Therefore, we can apply
  \reflem{lem:local-reduction-strong} with $W = \closure[\Delta]{m'}$
  and $U = \closure[\Delta]{m}$ to conclude that if $m'$ is reducible
  modulo ${\langle W \rangle} = {\langle S(m')\rangle}$ if and only if
  it is also reducible modulo
  $\langle U \rangle = \langle S(m)\rangle$.  Note that we use the
  fact that all vertices in $G[V\setminus \badset_\Delta]$---and hence
  also all vertices in $G[U\setminus W]$---have degree at most
  $\Delta$.

  With the \satcondition and reducibility conditions of
  \reflem{lem:ar-method} in hand, we conclude that every polynomial
  calculus refutation of $\Col{G}{k}$ requires degree strictly greater
  than $D$, as desired.
\end{proof}

\subsection{The Reducibility Lemma}

It remains to prove the reducibility lemma which we restate here for convenience.

\begin{restatablelem}{\ref{lem:local-reduction-strong}}[Reducibility
  Lemma, restated]
  Let $G=(V, E)$ be a $(\ell, 1/3\Delta)$\nobreakdash-sparse graph
  with a linear order $\vertexorder$ on $V$ and consider an admissible
  order that respects~$\vertexorder$. If the vertex sets
  $W\subseteq U$ satisfy hat $W$ is closed, that the size of $U$ is
  $\setsize{U} \leq \ell$, and that every vertex in
  $\nbhstd{U\setminus W}{W}$ has degree at most $\Delta$ in
  $G[U\setminus W]$, then for every monomial $\pcmonm$ such that
  $V(m) \subseteq W$, it holds that $\pcmonm$ is reducible modulo
  $\langle U\rangle$ if and only if $\pcmonm$ is reducible modulo
  $\langle W \rangle$.
\end{restatablelem}

The proof idea is to construct a function $\rho$ mapping variables
associated with vertices in $U \setminus W$ to either constants or
polynomials of smaller order such that all axioms
in~${\langle U \rangle \setminus \langle W\rangle}$ are either
satisfied or mapped to a polynomial in $\langle W\rangle$. It is not
hard to show that such a mapping turns any polynomial
in~$\langle U \rangle$ with leading monomial $m$ into a smaller
polynomial in~$\langle W \rangle$ whose leading monomial is also
$m$. It then follows that a monomial $m$ is reducible modulo
$\langle U \rangle$ if $m$ is reducible modulo $\langle W\rangle$. The
other direction is immediate, so this suffices to prove the lemma.

Let us first outline the construction of $\rho$.  Using the definition
of closure, we show in \reflem{lem:colouring-contraction} that there
exists a proper $3$\nobreakdash-colouring $\chi$ of the subgraph
${G[U\setminus (W \union \nbhstd{U}{W})]}$ that uses a \emph{single}
colour for each set $\nbhstd{U\setminus W}{u}$, where
$u \in \nbhstd{U\setminus W}{W}$.  Variables far from~$W$, which here
means variables associated with a vertex
in~${U \setminus (W \cup N_U(W))}$, are mapped according to the
$3$\nobreakdash-colouring $\chi$.  It remains to define $\rho$ on
variables associated with each vertex $u \in N_{U\setminus
  W}(W)$. Since~$u$ has precisely one adjacent vertex $v$ in~$W$, and
since furthermore the set $\nbhstd{U\setminus W}{u}$ is coloured with
a single colour, no matter how the vertex~$v$ is coloured there is
always a colour $c_u$ available to properly colour $u$. We may think
of $c_u$ as a function that, given $\chi$ and the colour of $v$,
assigns a colour to~$u$ that is consistent with the colouring of
$N_{U}(u)$.
Variables associated with the vertex~$u$ are mapped according to~$c_u$
by~$\rho$.

\begin{proof}[Proof of \reflem{lem:local-reduction-strong}]
  Since $W$ is a subset of $U$, it follows that if $m$ is reducible
  modulo $\langle W\rangle$, then $m$ is also reducible modulo
  $\langle U\rangle$. For the reverse direction, we define a mapping
  $ \rho$ on variables as outlined above.
      
  To this end, let $\chi$ be a proper $3$\nobreakdash-colouring of the
  subgraph ${G[U\setminus (W \union \nbhstd{U}{W})]}$ that uses a
  {single} colour for each set $\nbhstd{U\setminus W}{u}$, where
  $u \in \nbhstd{U\setminus W}{W}$. Such a colouring exists by
  \reflem{lem:colouring-contraction}.  Variables associated with a
  vertex~$u \in U \setminus \bigl(W \cup \nbhstd{U}{W}\bigr)$ are
  mapped according to $\chi$: if $\chi(u) = i$, then
  $\rho(x_{u,i}) = 1$ and $\rho(x_{u,i'}) = 0$ for all $i' \neq i$.
    
  Next, for each vertex $u\in \nbhstd{U\setminus W}{W}$, we define
  $\rho$ on the variables associated with $u$. Since $\chi$ assigns,
  for each $u \in N_{U \setminus W}(W)$, a single colour to each set
  $\nbhstd{U\setminus W}{u}$ there are at least two distinct colours
  $c_1, c_2 \in [k]$ that are not assigned to any vertex in
  $\nbhstd{U\setminus W}{u}$.  Since there are no $2$\nobreakdash-hops
  in $U$ with respect to $W$ the vertex~$u$ has a
  single
  neighbour $v\in W$. Furthermore, as there are no $3$-hops in $U$
  with respect to $W$, it holds that
  $N_{U \setminus W}(u) \cap N_{U\setminus W}(W) = \emptyset$, which
  implies that the vertex $v$ is the only neighbour of $u$ that is not
  coloured by $\chi$. Hence no matter how $v$ is coloured by $\chi$,
  either $c_1$ or $c_2$ can be used to properly colour $u$. Let us
  make this choice explicit by defining~$\rho$ on $u$ by
  \begin{align}
    \label{eq:substitution}
    \rho(x_{u,c}) =
    \begin{cases}
      x_{v, c_2}
      &\text{if $c = c_1$},\\
      \sum_{
      i \in [k],
      i \neq c_2
      }
      x_{v,i}
      &\text{if $c = c_2$, and}\\
      0
      &\text{otherwise, that is, if $c \not\in \set{c_1, c_2}$.}
    \end{cases}
  \end{align}
  This completes the definition of $\rho$. Note that the
  mapping~$\rho$ extends any proper $k$\nobreakdash-colouring of $W$
  to a proper $k$\nobreakdash-colouring of $U$.
    
  Let $\polyinu$ be a polynomial in~$\langle U \rangle$ with leading
  monomial~$\pcmonm$. We claim that
  \begin{enumerate}
  \item $\restrict{\polyinu}{\rho} \in \langle W\rangle$,
  \item that all monomials $m'$
    satisfy~$\restrict{m'}{\rho}\preceq m'$, and
  \item that~$m = \restrict{m}{\rho}$.
  \end{enumerate}
  If we can show this, then we are done, since $m$ is then the leading
  monomial of the
  polynomial~${\restrict{\polyinu}{\rho}\in \langle W \rangle }$ and
  we may thus conclude that if~$m$ is reducible
  modulo~$\langle U\rangle$, then $m$ is also reducible
  modulo~$\langle W\rangle$.

  We now argue that the three properties hold.  The latter two are
  almost immediate: since~$\rho$ does not map variables associated
  with $W$ (of which $V(m)$ is a subset) we
  have~$m = \restrict{m}{\rho}$.  Furthermore, since $\Desc{W} = W$
  and since $\prec$ is admissible, it holds for every variable $x$
  that~$\restrict{x}{\rho} \preceq x$, and hence every monomial~$m'$
  satisfies $\restrict{m'}{\rho} \preceq m'$.

  It remains to prove that
  $\restrict{\polyinu}{\rho} \in \langle W\rangle$. Since
  $\polyinu\in \langle U \rangle$ we may
  write~$\polyinu = \sum_i \polyinu_i p_i$ for polynomials
  $\polyinu_i \in \F[X]$ and axioms $p_i \in \Col{G[U]}{k}$. Note that
  the mapping~$\rho$ extends any proper $k$\nobreakdash-colouring
  of~$W$ to a proper $k$\nobreakdash-colouring of~$U$, and thus it
  follows by \reflem{lem:implication-ideal} that every axiom
  $p_i \in \Col{G[U]}{k}$ satisfies
  $\restrict{p_i}{\rho} \in \langle W\rangle$.  We can therefore
  conclude that the
  polynomial~${\restrict{\polyinu}{\rho} = \sum_i
    \restrict{\polyinu_i}{\rho}\cdot\restrict{p_i}{\rho}}$ is
  in~$\langle W \rangle $ as claimed.
\end{proof}

With the proof of \reflem{lem:local-reduction-strong} completed we
have shown the last missing piece of the proof of
\refthm{th:3-colouring-hard-sparse}. This thus establishes our
polynomial calculus degree lower bounds for the colouring formula over
sparse graphs.

\section{Concluding Remarks}
\label{sec:conclusion}

In this
work, we show that polynomial calculus over any field requires linear
degree to refute that a sparse random regular graph or \erdosrenyi
random graph is $3$-colourable.  Our lower bound is optimal up to
constant factors, and implies strongly exponential size lower bounds
by the well-known size-degree relation for polynomial
calculus~\cite{IPS99LowerBounds}.

Our overall proof technique is the same as that of earlier papers such
as
\cite{AR03LowerBounds,GL10Automatizability,GL10Optimality,MN15GeneralizedMethodDegree},
but our proofs have a different flavour. A central technical concept
in \cite{AR03LowerBounds,MN15GeneralizedMethodDegree} is (variations
of) the so-called \emph{constraint-variable incidence graph}: this
graph consists of a vertex per constraint and variable, and has an
edge between a constraint $C$ and a variable $x$ if and only if $C$
depends on $x$. This graph is commonly used to argue that by expansion
small sets of constraints are satisfiable, even after the removal of a
closed set of vertices. By contrast, we never need to make any
(explicit) use of this graph. This raises the question of whether it
is possible to rephrase our proofs in language closer to that of
\cite{AR03LowerBounds, MN15GeneralizedMethodDegree}, or are the two
approaches inherently different?

The lower bound techniques in this paper, as well as those in
\cite{MN15GeneralizedMethodDegree}, work over any field. For
$\NP$-hard problems such as $k$-colourability, we expect a polynomial
calculus lower bound to hold regardless of which field is used for the
derivations. However, other formulas such as the Tseitin
contradictions are easy for polynomial calculus over a field of
characteristic 2 and hard in other characteristics. The techniques in,
\eg,~\cite{BI99RandomJOURNALREF,BGIP01LinearGaps,AR03LowerBounds}
capture this fact, while those in~\cite{MN15GeneralizedMethodDegree}
and this paper cannot. Another interesting question is therefore
whether the techniques in these papers can be unified into a general
approach that works both for field-dependent and field-independent
lower bounds.

Our degree lower bounds for $3$-colourability are of the form
$n/f(d)$, where $d$ is either degree or average degree of the graph
depending on the random graph model. In our work, $f$ is at least
exponential in~$d$, but in previous results
\cite{BCCM05RandomGraph,LN17GraphColouring}, $f$ is at most polynomial
in~$d$. While the precise dependence on $d$ is immaterial for sparse
random graphs, it would be interesting to see if the parameters in our result can be
improved. We remark that it is far from clear what the correct
dependence on $d$ should be.
For the sums-of-squares proof system, which simulates polynomial
calculus over the reals~\cite{Berkholz18Relation}, there exist strong
upper bounds for $k$-colourability on random graphs and random regular
graphs in some parameter regimes: the paper~\cite{JKM19Lovasz} showed
that \aas, degree-$2$ sums\nobreakdash-of\nobreakdash-squares refutes
$k$-colourability on $d$-regular random graphs if ${d \geq
  4k^2}$. These results rule out a polynomial dependence on $d$ in any linear
sums\nobreakdash-of\nobreakdash-squares degree lower bound for
$k$-colourability whenever $k$ is fixed. However, similar upper bounds are
not known to hold for polynomial calculus, and it should be pointed
out that the latter proof system is incomparable to sum-of-squares
when considered over fields of finite characteristic.

More broadly it would be interesting to investigate whether
the ideas and concepts underlying this work could be extended
to prove lower bounds for colouring principles in
other proof systems, the most obvious candidates being
Sherali-Adams and sums-of-squares.
Regarding polynomial calculus, it is worth noting that the closure
operation defined in~\cite{RT22GraphsLargeGirth} and generalized in
this work is not, per se, restricted to graph colouring. It is natural
to ask whether similar techniques could be useful for proving degree
lower bounds for other graph problems.  One open problem is to improve
the degree lower bound for matching on random graphs
in~\cite{AR22PerfectMatchingJournal} to linear in the graph size, and
to make it hold for graphs of small constant degree. Another problem
is to establish polynomial calculus size lower bounds for independent
set and vertex cover, analogously to what was done for the resolution
proof system in~\cite{BIS07IndependentSets}. Finally, an intriguing
technical challenge is to prove degree lower bounds for variants of
the dense linear ordering
principle~\cite{AD08CombinatoricalCharacterization} for graphs of
bounded degree.

\section*{Acknowledgements}

The authors would like to thank Albert Atserias, Gaia Carenini and
Amir Yehudayoff for helpful discussions during the course of this
work, and we also thank Albert for making us aware of some relevant
references.  In addition, we benefitted from feedback of the
participants of the Dagstuhl Seminar~23111 ``Computational Complexity
of Discrete Problems'' and Oberwolfach workshop~2413 ``Proof
Complexity and Beyond''.  Finally, we are grateful to Maryia Kapytka
for feedback on preliminary versions of this manuscript and also to
the anonymous FOCS reviewers---all these comments helped us improve
the exposition in the paper considerably.

Part of this work was carried out while the authors were taking part
in the semester programme \emph{Meta-Complexity} and the extended
reunion of the programme \emph{Satisfiability: Theory, Practice, and
  Beyond} at the Simons Institute for the Theory of Computing at UC
Berkeley in the spring of 2023.

Susanna F. de Rezende received funding from ELLIIT, from the Knut and
Alice Wallenberg grant \mbox{KAW 2021.0307}, and from the Swedish Research
Council grant \mbox{2021-05104}.
Jonas Conneryd and Jakob Nordström were funded by
the Swedish Research Council grant \mbox{2016-00782},
and in addition Jonas Conneryd was also 
partially supported by the
Wallenberg AI, Autonomous Systems and Software Program (WASP)
funded by the Knut and Alice Wallenberg Foundation,
whereas Jakob Nordström together with 
Shuo Pang were supported by the
Independent Research Fund Denmark grant \mbox{9040-00389B}.
Kilian Risse was supported by the Swiss National Science Foundation
project \mbox{200021-184656}
``Randomness in Problem Instances and Randomized Algorithms''.

\appendix 

\section{On Boolean Implication and Ideal Membership}
\label{sec:polyclaim}

In this appendix, we provide a proof of the folklore result
that being implied by a set of polynomials and being in the ideal
generated by these polynomials is the same in the Boolean setting.

\begin{restatablelem}{\ref{lem:implication-ideal}}[restated]
  Let $g$ be a polynomial and $Q$ be a set of polynomials
  in~$\F[x_1, \ldots, x_n]$, and suppose that $Q$ contains all the
  Boolean axioms. Then it holds that $g$ vanishes on all common roots
  of~$Q$ if and only if $g \in \langle Q \rangle$.
\end{restatablelem}

\begin{proof}
  If $g \in \langle Q \rangle$, then we can write
  $g = \sum_{i} f_i q_i$ for $f_i \in \F[x_1, \ldots, x_n]$ and
  $q_i \in Q$. Observe that $\restrict{\sum_{i} f_i q_i}{\xi} = 0$ for
  any common root $\xi \in \set{0,1}^n$ of $Q$. Hence $g$ vanishes on
  all common roots of $Q$.

  For the other direction, let $\xi = (\xi_1, \xi_2, \ldots, \xi_n)$
  be an element of $\{0,1\}^n$ and write $\indic_\xi$ for the
  multilinear polynomial that evaluates to $1$ on $\xi$ and to $0$ on
  all other elements of $\{0,1\}^n$, that is,
  \begin{equation}
    \indic_\xi(x)
    =
    \prod_{i: \xi_i = 1}x_i
    \prod_{j: \xi_j = 0}(1-x_j)
    \eqperiod
  \end{equation} 
  Clearly, every function $f\colon \{0, 1\}^n \to \F$ can be expressed
  as a multilinear polynomial through the identity
  \begin{equation}
    f
    =
    \sum_{\xi \in \{0, 1\}^n}
    f(\xi)
    \cdot
    \indic_\xi
    \eqcomma
  \end{equation}
  and as the polynomials $\set{\indic_\xi\mid \xi \in \{0, 1\}^n}$
  form a basis of the vector space of multilinear polynomials
  over~$\F[x_1, \ldots, x_n]$, this representation is unique. Let
  $S\subseteq \{0, 1\}^n$ be the set of common roots of the
  polynomials in~$Q$. Since $g$ vanishes on $S$, we may write
  \begin{equation}
    g
    =
    \sum_{\xi \in \{0, 1\}^n\setminus S}
    g(\xi)
    \cdot
    \indic_\xi
    \eqperiod
  \end{equation}
  We now show that if $\xi\in \{0, 1\}^n \setminus S$, then it holds
  that $\indic_\xi\in \langle Q \rangle$. This suffices to prove the
  desired result, as $g$ is then a linear combination of
  polynomials in~$\langle Q \rangle$ and hence is also in
  $\langle Q \rangle$. 
    
  Let $\xi$ be an element in~$\{0, 1\}^n \setminus S$. Because $S$ is
  the set of common roots of the polynomials in~$Q$, there exists a
  polynomial $q \in Q$ such that $q(\xi) \neq 0$. The polynomial
  $\indic_\xi\cdot (q(\xi))^{-1}q$
  coincides with $\indic_\xi$ on all of $\{0, 1\}^n$, so $\indic_\xi\cdot (q(\xi))^{-1}q = \indic_\xi$ modulo the Boolean axioms. 
  Since~$\langle Q \rangle$ contains both~$\indic_\xi\cdot (q(\xi))^{-1}q$ and the Boolean axioms, it
  follows that $\indic_\xi \in \langle Q \rangle$.
\end{proof}
 
\section{Random Graphs Are Sparse}
\label{sec:sparsity}

In this appendix we prove \reflem{lem:razborov-sparse}, which is a
quantitative version of Lemma 4.15 in
\cite{Razborov17WidthSemialgebraic}. We make no claim of originality,
but present this result here to make the paper self-contained.

We start by proving the sparsity lemma for graphs sampled from the
\erdosrenyi random graph distribution~$\gndn$ in
\reflem{lem:razborov-sparse-Gnp} and then establish the analogous
result for random regular graphs sampled according to~$\gnd$ in
\reflem{lem:razborov-sparse-Gnd}.

\begin{lemma}[Sparsity lemma for \erdosrenyi random graphs]
  \label{lem:razborov-sparse-Gnp}
  Let $n,d\in \N^+$ and $\epsilon, \delta \in \R^+$ be such that
  $\epsilon \delta = \omega(1/\log n)$.  If $G$ is a graph sampled
  from~$\gndn$, then \aas it is
  $((4d)^{- (1+\delta)(1+\epsilon)/\epsilon}n, \epsilon)$-sparse.
\end{lemma}
\begin{proof}
  Let $\alpha = (4d)^{- (1+\delta)(1+\epsilon)/\epsilon}$ and denote
  by $\eventa$ the event ``$G$ is $(\alpha n, \epsilon)$-sparse''.
  For a set $U\subseteq V$ of size $s$, the random variable
  $|E(G[U])|$ is a sum of $s(s-1)/2$ random indicator variables for
  the edges that are $1$ with probability $d/n$ and $0$ otherwise. We
  apply a union bound over sets of size $s\leq \alpha n$ to conclude
  that
  \begin{subequations}
    \begin{align}
      \Pr[\lnot \eventa]  
      &\leq
        \sum_{\substack{U \subseteq V\\  \setsize{U} \leq \alpha n}}
        \Pr[\setsize{E(G[U])} \geq (1+\epsilon) \setsize{U}] \\
      &\le
        \sum_{s=1}^{\alpha n}
        \binom{n}{s}
        \binom{\frac{s(s-1)}{2}}{(1+\epsilon) s} \cdot
        \left(\frac{d}{n}\right)^{(1+\epsilon) s}   \\
      &\le
        \sum_{s=1}^{\alpha n}
        \left(\frac{en}{s}\right)^s
        \left(\frac{e(s-1)}{2(1+\epsilon) }\right)^{(1+\epsilon) s}
        \cdot
        \left(\frac{d}{n}\right)^{(1+\epsilon) s}  \\
      &\le
        \sum_{s=1}^{\alpha n}
        \exp
        \left(
        - \epsilon s \ln\left(\frac{n}{s}\right) +
        (1+\epsilon) s
        \left(
        \ln
        \left(\frac{e^2d}{2(1+\epsilon)}\right)
        \right)
        \right) \\
      &\le
        \sum_{s=1}^{\alpha n}
        \exp \left( -
        {\delta\epsilon s \ln(n/s)}
        \right) \label{eq:bound-alpha} \\
      &\leq
        o(1)
        \label{eq:bound-epsilon}
        \eqcomma
    \end{align}
  \end{subequations} 
  where for \refeq{eq:bound-alpha} we use that
  $n/s \geq 1/\alpha = (4d)^{(1+\delta)(1+\epsilon)/\epsilon}$ to
  estimate that
  \begin{align}
    (1+\epsilon)
    \left(\ln\left(\frac{e^2d}{2(1+\epsilon)}\right)\right)
    \le
    (1+\epsilon)
    \ln\left({4d}\right)
    =
    \frac{\epsilon  \ln(1/\alpha)}{1+\delta}
    \le
    \frac{\epsilon  \ln(n/s)}{1+\delta}
    \eqcomma
  \end{align}
  and for \refeq{eq:bound-epsilon} we use that
  $\epsilon \delta = \omega(1/\log n)$ and that
  $s\ln(n/s) \geq \ln n + s-1 $ (for $1 \leq s\leq n/e^2$).
\end{proof}

We next prove the sparsity lemma for random regular graphs.  In order
to sample $G$ from $\gnd$ we use the \emph{configuration model}, which
is defined as follows.  Given $n$ and $d$ such that $dn$ is even, we
have a vertex set $V$ of size $n$ and for each vertex $v\in V$ there
is a cell $C_v$ with $d$ elements.  We sample a perfect matching $M$
uniformly from the set $\mathcal{M}_{dn}$ of all possible perfect
matchings of the $dn$ elements and consider the corresponding
multi-graph $G_M = (V,E)$, possibly with parallel edges and loops,
where $(u,v)\in E$ if and only if there exists $(x,y)\in {M}$ such
that $x\in C_u$ and $y\in C_v$.  Since all simple graphs (without
parallel edges or loops) are sampled with the same probability, we can
sample $G$ from $\gnd$ by sampling according to the configuration
model repeatedly until we sample a simple graph.
\begin{theorem}[\cite{MW91AsymptoticEnumeration,Wormald99ModelsRandom}]
  For $d = o(n^{1/2})$, the probability that $G_M$ is simple when $M$
  is sampled uniformly from $\mathcal{M}_{dn}$ is equal to
  \begin{equation*}
    \exp\left(
      - \frac{{d^2 - 1}}{4}
      - \frac{d^3}{12 n}
      + O(d^2/n)
    \right) \eqperiod
  \end{equation*}
\end{theorem}

Let $S_{\ell, q}$ denote the sum of $\ell$ random variables that are
$1$ with probability $q$ and $0$ otherwise.  We argue that we can
bound the probability that a set of vertices $U\subseteq V$ witnesses
that the graph is not sparse by bounding the probability that
$S_{\ell, q}$ is too large.
\begin{claim}\label{claim:Bernoulli}
  For any $s\le n/2$ and $B\in \R^+$, if $U\subseteq V$ is of size $s$
  and $q = s/(n - s)$, it holds that
  \begin{equation*}
    \Pr_{M\sim \mathcal{M}_{dn}}
    [
    \setsize{E(G_M[U])} \geq B
    ]
    \le
    \Pr[S_{ds, q} \geq B] \eqperiod
  \end{equation*}
\end{claim}
\begin{proof}
  To see why this is true, consider the random process in the
  configuration model that matches one by one the elements in the
  cells $C_v$ for all $v\in U$. At each step, there are at most $ds$
  elements in cells $C_v$ where $v\in U$ that are not yet matched, and
  at least $d(n-s)$ elements in cells $C_v$ where $v\not\in U$ that
  are not yet matched.  This implies that at each step the probability
  that we obtain an edge between cells in $U$ is at most
  $ds/d(n-s)=q$. Since at least one element in a cell of $U$ is
  matched at every step, there are at most $ds$ steps in total. Hence
  the claim follows.
\end{proof}

\begin{lemma}[Sparsity lemma for random regular graphs]
  \label{lem:razborov-sparse-Gnd}
  Let $n,d\in \N^+$ and $\epsilon, \delta \in \R^+$ be such that
  $\epsilon \delta = \omega(1/\log n)$ and
  $d^2\le \epsilon \delta \log n$.  If $G$ is a graph sampled
  from~$\gnd$, then \aas it is
  $((8d)^{- (1+\delta)(1+\epsilon)/\epsilon}n, \epsilon)$-sparse.
\end{lemma}

\begin{proof}
  Fix $\epsilon >0$, let $\alpha=(8d)^{-2(1+1/\epsilon)}$, and denote
  by $\eventa$ the event ``$G$ is $(\alpha n, \epsilon)$-sparse''.  We
  want to prove that $G\sim \gnd$ is $(\alpha n, \epsilon)$-sparse
  with probability that goes to $1$ as $n$ goes to infinity.  To this
  end, we prove that if we sample $G$ from the configuration model,
  the probability that it is not $(\alpha n, \epsilon)$-sparse is much
  smaller than the probability that $G$ is a random regular
  graph. More formally, our goal is to prove that
  \begin{align}
    \Pr[\lnot \eventa]
    \cdot
    \exp\left(\frac{d^2 - 1}{4}\right)
    \leq o(1) \eqcomma
  \end{align}
  where we recall that the probability here and in what follows is
  taken over sampling $G$ in the configuration model.  By union bound
  and using \refclaim{claim:Bernoulli} we have that
  \begin{subequations}
    \begin{align}
      \Pr[\lnot \eventa]
      \cdot
      \exp\left(\frac{d^2 - 1}{4}\right)
      &\leq
        \sum_{\substack{U \subseteq V\\ \setsize{U} \leq \alpha n}}
        \Pr[\setsize{E(G[U])} \geq (1+\epsilon) \setsize{U}]
        \cdot
        \exp\left(\frac{d^2 }{4}\right) \\
      &\le
        \sum_{s=1}^{\alpha n}
        \binom{n}{s}
        \Pr[S_{ds, q} \geq (1+\epsilon) s]
        \cdot
        \exp\left(\frac{d^2 }{4}\right) \eqperiod
    \end{align}
  \end{subequations} 
  Our goal is to show that $Pr[S_{ds, q} \geq (1+\epsilon) s]$ is so
  small that it compensates for the other factors. More concretely, we
  wish to show that for $s\leq \alpha n$ it holds that
  \begin{align}\label{eq:smallprob}
    \binom{n}{s}
    \Pr[S_{ds, q} \geq (1+\epsilon) s]  
    &\le
      \exp
      \left(
      -
      {\epsilon \delta s \ln(n/s)}
      \right) \eqcomma
  \end{align}
  from which it follows, since
  $d^2 \le \epsilon \delta \log n \le 2\epsilon \delta \ln n $, that
  \begin{subequations}
    \begin{align}
      \Pr[\lnot \eventa]
      \cdot
      \exp\left(\frac{d^2 - 1}{4}\right)
      &\le
        \sum_{s=1}^{\alpha n}
        \exp
        \left(
        -{\epsilon\delta s \ln(n/s)}
        + \frac{d^2}{4}
        \right) \\
      &\le
        \sum_{s=1}^{\alpha n}
        \exp\left(
        -\frac{\epsilon \delta s \ln(n/s)}{2}
        \right) \\
        &= o(1) \eqcomma
    \end{align}
  \end{subequations} 
  where we use the fact that $\epsilon \delta = \omega(1/\log n)$ and
  that $s\ln(n/s) \geq \ln n + s - 1 $ (for $1 \leq s\leq n/e^2$).

  It remains to show that \refeq{eq:smallprob} holds. This follows
  from the sequence of calculations
  \begin{subequations}
    \begin{align}
      \binom{n}{s}
      \Pr[S_{ds, q} \geq (1+\epsilon) s]
      &\le
        \binom{n}{s}
        \binom{ds}{(1+\epsilon) s}
        \cdot
        \left(\frac{s}{n-s}\right)^{(1+\epsilon) s}   \\
      &\le
        \left(\frac{en}{s}\right)^s
        \left(\frac{eds}{(1+\epsilon)s}\right)^{(1+\epsilon) s}
        \cdot
        \left(\frac{s}{n-s}\right)^{(1+\epsilon) s}  \\
      &\le
        \exp
        \left(-
        \epsilon s \ln \left(\frac{n-s}{s}\right) +
        s +
        (1+\epsilon) s
        \ln\left(\frac{ed}{(1+\epsilon)}\right)
        \right)   \\
      &\le
        \exp \left(-
        \epsilon s \ln \left(\frac{n}{s}\right) +
        (1+\epsilon) s
        \ln\left(\frac{e^2d}{(1+\epsilon)}\right)
        \right) \label{eq:bound-ln}  \\
      &\le
        \exp
        \left(
        - {\epsilon \delta s \ln \left({n}/{s}\right)}
        \right) \label{eq:bound-alphaGnd}
        \eqcomma
    \end{align}
  \end{subequations} 
  where for \refeq{eq:bound-ln} we use that
  $\ln(n/s - 1) \geq \ln(n/s) -1$ (for $s \leq n/2$), and for
  \refeq{eq:bound-alphaGnd} we use that
  $n/s \geq 1/\alpha = (8d)^{-(1+\delta)(1+\epsilon)/\epsilon}$ to
  bound
  \begin{align}
    (1+\epsilon)
    \left(
    \ln\left(\frac{e^2d}{(1+\epsilon)}\right)
    \right)
    \le
    (1+\epsilon)
    \ln\left({8d}\right)
    =
    \frac{\epsilon  \ln(1/\alpha)}{1+\delta}
    \le
    \frac{\epsilon  \ln(n/s)}{1+\delta} \eqperiod
  \end{align}
  This concludes the proof of \reflem{lem:razborov-sparse}.
\end{proof}

\section{\erdosrenyi Graphs Almost Have Small Maximum Degree}
\label{sec:TDelta-small}

In this section we provide a proof of \reflem{lem:TDelta-small-v2}
stating that \erdosrenyi random graphs have small degree except for a
small set of vertices.

\begin{restatablelem}{\ref{lem:TDelta-small-v2}}[restated]
  Let~$G=(V, E)$ be a graph sampled from~$\gndn$ where
  $d=\bigoh{\log n}$.  If $\Delta\geq d$ is such that
  $(\Delta/ed)^\Delta = \littleoh{n}$, then \aas there exists a set
  $\badset_\Delta$ of size at most $(ed/\Delta)^\Delta \cdot 2en$ such
  that the maximum degree in~$G[V\setminus \badset_\Delta]$ is at
  most~$\Delta-1$.
\end{restatablelem}
\begin{proof}
  Let~$\eventa$ denote the event that there exists a
  set~$\badset_\Delta \subseteq V$ of
  size~$\ell = (ed/\Delta)^\Delta \cdot 2en$ such that the maximum
  degree in~$G[V\setminus \badset_\Delta]$ is at most~$\Delta-1$. We
  prove that~$\Pr[\lnot \eventa]= \littleoh{1}$.
    
  Note that if $\eventa$ does not hold, then, in particular, if we go
  over the vertices of~$G$ in any given fixed order, and remove from
  $G$ any vertex of degree at least $\Delta$ that we encounter, we
  will end up removing at least $\ell$ vertices. Observe further that
  after the removal of $i$ vertices, the probability that a vertex has
  degree at least $\Delta$ is at most
  $\binom{n -i}{\Delta} (d/n)^\Delta$ and is independent of the fact
  that the removed vertices had degree at least $\Delta$.  Therefore,
  by taking a union bound over all sets of size $\ell$, we can bound
  the probability of the event $\eventa$ not holding by
  \begin{subequations}
    \begin{align}
      \Pr[\lnot \eventa] 
      &\leq
        \binom{n}{\ell}
        \prod_{i = 1}^\ell
        \binom{n-i}{\Delta}
        \left(\frac{d}{n}\right)^\Delta \\
      &\leq
        \binom{n}{\ell}
        \left(
        \binom{n}{\Delta}
        \left(
        \frac{d}{n}
        \right)^\Delta
        \right)^\ell \\
      &\leq
        \left(
        \frac{en}{\ell}
        \right)^\ell
        \left(
        \frac{ed}{\Delta}
        \right)^{\Delta\ell}\\
      &= \littleoh{1} \eqcomma \label{eq:last}
    \end{align}
  \end{subequations}
  where for \refeq{eq:last} we use that
  $\ell = 2en \left({ed}/{\Delta}\right)^{\Delta}$ to derive that
  $\left({en}/{\ell}\right)^\ell
  \left({ed}/{\Delta}\right)^{\Delta\ell} = 2^{-\ell}$, and then use
  that $\ell = \littleomega{1}$, which holds since
  $(\Delta/ed)^\Delta = \littleoh{n}$, to conclude
  $2^{-\ell} = \littleoh{1}$. The lemma follows.
\end{proof}

\bibliography{refpaper}

\bibliographystyle{alpha}

\end{document}